\documentclass{elsarticle}
\usepackage{amssymb}
\usepackage{latexsym}
\usepackage{mathtools}
\usepackage{tikz}
\usetikzlibrary{calc,petri}
\usetikzlibrary{snakes}
\usetikzlibrary{patterns}
\usetikzlibrary{external}
\newtheorem{theorem}{Theorem}
\newtheorem{lemma}[theorem]{Lemma}
\newtheorem{corollary}[theorem]{Corollary}
\newdefinition{definition}[theorem]{Definition}
\newdefinition{example}[theorem]{Example}
\newtheorem{remark}[theorem]{Remark}
\newproof{Proof}{Proof}
\newenvironment{proof}{\begin{Proof}}{\qed\end{Proof}}

\newcommand{\sdif}{\mathop{\mathrm{\Delta}}}
\renewcommand{\emptyset}{\varnothing}

\begin{document}

\begin{frontmatter}

\title{Sorting by Reversals and the Theory of 4-Regular Graphs}

\author{Robert Brijder\fnref{myfootnote}}
\address{Hasselt University, Belgium}
\fntext[myfootnote]{Postdoctoral fellow of the Research Foundation -- Flanders (FWO).}
\ead{robert.brijder@uhasselt.be}

\begin{abstract}
We show that the theory of sorting by reversals fits into the well-established theory of circuit partitions of 4-regular multigraphs (which also involves the combinatorial structures of circle graphs and delta-matroids). In this way, we expose strong connections between the two theories that have not been fully appreciated before. 
We also discuss a generalization of sorting by reversals involving the double-cut-and-join (DCJ) operation. Finally, we also show that the theory of sorting by reversals is closely related to that of gene assembly in ciliates.
\end{abstract}


\end{frontmatter}

\section{Introduction} \label{sec:introd}
Edit distance measures for genomes can be used to approximate evolutionary distance between their corresponding species. A number of genome transformations have been used to define edit distance measures. In this paper we consider the well-studied chromosome transformation called \emph{reversal}, which is an inversion of part of a chromosome \cite{PevznerBook}. If two given chromosomes can be transformed into each other through reversals, then the difference between these two chromosomes can be represented by a permutation, where the identity permutation corresponds to equality. As a result, transforming one chromosome into the other using reversals is called \emph{sorting} by reversals. The \emph{reversal distance} is the \emph{least} number of reversals needed to accomplish this transformation (i.e., to sort the permutation by reversals). In \cite{DBLP:journals/jacm/HannenhalliP99} a formula is given for the reversal distance, leading to an efficient algorithm to compute reversal distance. The proof of that formula uses a notion called the breakpoint graph. Subsequent streamlining of this proof led to the introduction of additional notions \cite{siamcomp/KaplanST99} such as the overlap graph and a corresponding graph operation. A reversal can be seen as a special case of a double-cut-and-join (DCJ) operation that can operate either within a chromosome or between chromosomes. Similar as for reversals, one can define a notion of DCJ distance, and a formula for DCJ distance is given in \cite{wabi/BergeronMS06/DCJformula}.

The theory of circuit partitions of 4-regular multigraphs was initiated in \cite{kotzig1968} and is currently well-developed with extensions and generalizations naturally leading into the domains of, e.g., linear algebra and matroid theory. In this paper we show that this theory can be used to study sorting by reversals, and more generally sorting by DCJ operations. Moreover, we show how various aspects of the theory of circuit partitions of 4-regular multigraphs relate to the topic of sorting by reversals. In particular, we show that the notion of an overlap graph and its corresponding graph operation from the context of sorting by reversals \cite{siamcomp/KaplanST99} correspond to circle graphs and the (looped) local complementation operation, respectively, of the theory of circuit partitions of 4-regular multigraphs. This leads to a reformulation of the Hannenhalli-Pevzner theorem \cite{DBLP:journals/jacm/HannenhalliP99} to delta-matroids. We remark however that the theory of circuit partitions of 4-regular multigraphs is too broad to fully cover here, and so various references are provided in the paper for more information.

It has been shown that various other research topics can also be fit into the theory of circuit partitions of 4-regular multigraphs: examples include the theory of ribbon graphs (or embedded graphs) \cite{Bouchet1989/maps_deltam,ChunMoffatt/DeltaM/EmbeddedGraphs} and the theory of gene assembly in ciliates \cite{BH/algebra-Tampa}. As such all these research topics arising from different contexts turn out to be strongly linked, and results from one research topic can often be carried over to another. Indeed, it is not surprising that the Hannenhalli-Pevzner theorem has been independently discovered in the context of gene assembly in ciliates \cite{SuccessfulnessChar_Original,BinarySymmetric/BrijderH12}.

This paper is organized as follows. In Section~\ref{sec:sort_revs} we recall sorting by reversals and in Section~\ref{sec:4reg_graph} we associate a 4-regular multigraph and a pair of circuit partitions to a pair of chromosomes (where one is obtainable from the other by reversals). This leads to a reformulation of a known inequality of the reversal distance in terms of 4-regular multigraphs. In Section~\ref{sec:circle_graphs} we associate a circle graph to the two circuit partitions and we show that the adjacency matrix representation of this circle graph reveals essential information regarding the reversal distance. We recall local complementation in Section~\ref{sec:lc} and the Hannenhalli-Pevzner theorem in Section~\ref{sec:HPthm}, and reformulate the Hannenhalli-Pevzner theorem in terms of delta-matroids in Section~\ref{sec:delta_m}. Before discussing delta-matroids, we also recall sorting by DCJ operations in Section~\ref{sec:DCJ_multiple}. We discuss the close connection of sorting by reversals and gene assembly in ciliates in Section~\ref{sec:ga_ciliates}. Finally, a discussion is given in Section~\ref{sec:disc}.

\section{Sorting by reversals} \label{sec:sort_revs}

In this section we briefly and informally recall notions concerning sorting by reversals. See, e.g., the text books \cite{PevznerBook,genomeRearr/Fertin/2009} for a more formal and extensive treatment.

\newcommand{\signp}[9]{
\begin{tikzpicture}[auto,x=0.8cm,y=0.8cm]
\draw (0,1) -- (7,1) ;
\draw (0,0) -- (7,0);
\foreach \x in {0,...,7}
{
  \draw (\x,0) -- (\x,1);
}
\draw (0,0) +(0.5,0.5) node{${#1}$};
\draw (1,0) +(0.5,0.5) node{${#2}$};
\draw (2,0) +(0.5,0.5) node{${#3}$};
\draw (3,0) +(0.5,0.5) node{${#4}$};
\draw (4,0) +(0.5,0.5) node{${#5}$};
\draw (5,0) +(0.5,0.5) node{${#6}$};
\draw (6,0) +(0.5,0.5) node{${#7}$};
\draw [
    thick,
    decoration={
        brace,
        mirror,
        raise=0.1cm
    },
    decorate
] ({#8},0) -- ({#9},0)
node {};
\end{tikzpicture}
}

\begin{figure}
\begin{center}
\resizebox{\textwidth}{!}{
\begin{tikzpicture}
\node (n1) at (0,0) {\signp{1}{-6}{7}{4}{-2}{-5}{3}{0}{0}};
\node (n2) at (8,0) {\signp{1}{2}{3}{4}{5}{6}{7}{0}{0}};
\node (m1) at (0,-1) {chromosome of some species $A$};
\node (m2) at (8,-1) {chromosome of some species $B$};
\end{tikzpicture}
}
\end{center}
\caption{Difference of two chromosomes of species $A$ and $B$. The difference is described with respect to the ordering of the chromosome segments of $B$.}
\label{fig:sign_perm}
\end{figure}

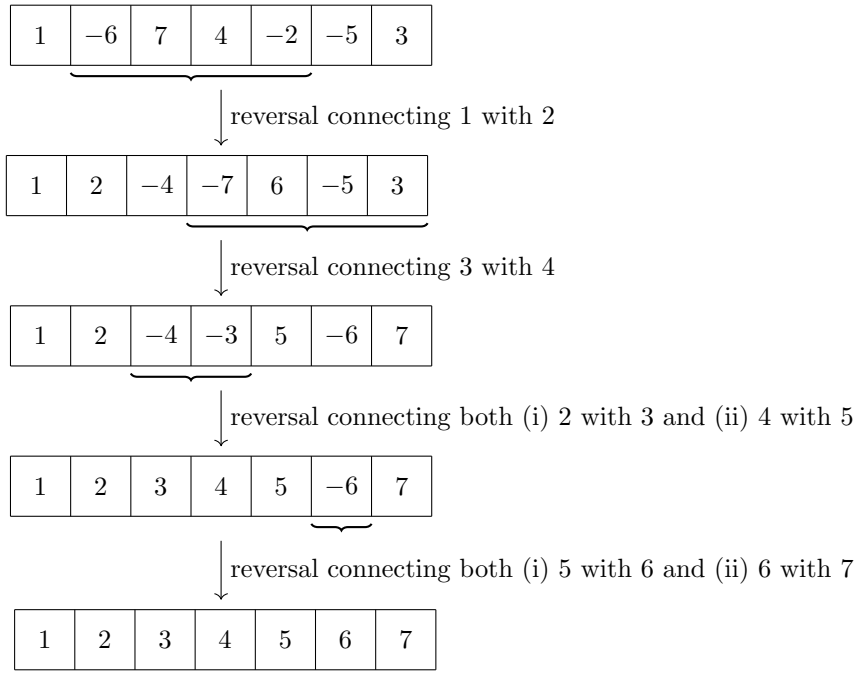
\begin{figure}
\begin{center}
\begin{tikzpicture}
\node (n1) at (0,8) {\signp{1}{-6}{7}{4}{-2}{-5}{3}{1}{5}};
\node (n2) at (0,6) {\signp{1}{2}{-4}{-7}{6}{-5}{3}{3}{7}};
\node (n3) at (0,4) {\signp{1}{2}{-4}{-3}{5}{-6}{7}{2}{4}};
\node (n4) at (0,2) {\signp{1}{2}{3}{4}{5}{-6}{7}{5}{6}};
\node (n5) at (0,0) {\signp{1}{2}{3}{4}{5}{6}{7}{0}{0}};
\draw[->] (n1) to node [right] {reversal connecting $1$ with $2$} (n2);
\draw[->] (n2) to node [right] {reversal connecting $3$ with $4$} (n3);
\draw[->] (n3) to node [right] {reversal connecting both (i) $2$ with $3$ and (ii) $4$ with $5$} (n4);
\draw[->] (n4) to node [right] {reversal connecting both (i) $5$ with $6$ and (ii) $6$ with $7$} (n5);
\end{tikzpicture}
\end{center}
\caption{Optimal sorting of the signed permutation of Figure~\ref{fig:sign_perm} by reversals.}
\label{fig:sign_perm_sort}
\end{figure}

During the evolution of species, various types of modifications of the genome may occur. One such modification is the inversion (i.e., rotation by 180 degrees) of part of a chromosome, called a \emph{reversal}. In Section~\ref{sec:DCJ_multiple}, we recall that this inversion is the result of a so-called double-cut-and-join operation. The \emph{reversal distance} between two given chromosomes is the minimal number of reversals needed to transform one into the other, and it is a measure of the evolutionary distance between the two species. Figure~\ref{fig:sign_perm} shows two toy chromosomes which have seven segments in common, but their relative positions and orientations differ (by, e.g., $-5$ we mean segment $5$ in inverted orientation, i.e., rotated by 180 degrees). Figure~\ref{fig:sign_perm_sort} shows that the reversal distance between the chromosomes of Figure~\ref{fig:sign_perm} is at most four.

We can concisely describe the chromosomes of Figure~\ref{fig:sign_perm} by the sequences $(1, -6, 7, 4, -2, -5, 3)$ and $(1, 2, 3, 4, 5, 6, 7)$. These sequences are called \emph{signed permutations} in the literature, and we will adopt this convention here (although we won't treat them as permutations in this paper). A signed permutation of the form $(1, 2, \cdots, n)$ is called the \emph{identity permutation}. The \emph{reversal distance} of a single signed permutation $\pi$, denoted by $d_r(\pi)$, is the minimal number of reversals needed to transform $\pi$ into the identity permutation. Thus, for $\pi = (1, -6, 7, 4, -2, -5, 3)$ of Figure~\ref{fig:sign_perm}, we have $d_r(\pi) \leq 4$ by Figure~\ref{fig:sign_perm_sort}.

Viewing the chromosome of species $B$ of Figure~\ref{fig:sign_perm} as the ``sorted'' chromosome, the transformation using reversals of the chromosome of species $A$ to the chromosome of species $B$ is called \emph{sorting by reversals}.

\begin{remark}\label{rem:rev_dist_one_diff}
For any signed permutation $\pi = (\pi_1, \pi_2, \cdots, \pi_n)$, the signed permutation $\overline{\pi} = (-\pi_n, -\pi_{n-1}, \cdots, -\pi_1)$ represents the same chromosome (just considered 180 degrees rotated). The reversal distance of $\pi$ and $\overline{\pi}$ may however differ (this difference is, of course, at most one because $\pi$ can be turned into $\overline{\pi}$ using one ``full'' reversal). Therefore, one could argue that a better notion of the reversal distance of $\pi$ would be the minimum value of the reversal distances of both $\pi$ and $\overline{\pi}$. Equivalently, one could view both $(1, 2, \cdots, n)$ and $(-n, -(n-1), \cdots, -1)$ as identity permutations. We revisit this issue in Sections~\ref{sec:4reg_graph} and \ref{sec:DCJ_multiple}.
\end{remark}

\section{Four-regular multigraphs} \label{sec:4reg_graph}

In this paper, \emph{graphs} are allowed to have loops but not multiple edges, and \emph{multigraphs} are allowed to have both loops and multiple edges. We denote the sets of vertices and edges of a (multi)graph $G$ by $V(G)$ and $E(G)$, respectively. For graphs, each edge $e \in E(G)$ is either of the form $\{v\}$ (i.e., $v$ is a looped vertex) or of the form $\{v_1,v_2\}$ (i.e., there is an edge between $v_1$ and $v_2$). A vertex $v$ is said to be \emph{isolated} if no edge is incident to it (in particular, $v$ is not looped). A \emph{4-regular multigraph} is a multigraph where each vertex has degree $4$, a loop counting as two.

A standard tool for the calculation of the reversal distance is the so-called breakpoint graph of a signed permutation. In this section we instead assign a 4-regular multigraph and two of its circuit partitions to a signed permutation. The main reason for considering this graph instead of the breakpoint graph is that in this way we can use the vast amount of literature concerning the theory of circuit partitions of 4-regular multigraphs, which began with the seminal paper of Kotzig \cite{kotzig1968}, and extended, e.g., in \cite{Abrham1980/EulerTours,Fleischner/EulerianTrails}.
We also note that a drawback of using the breakpoint graph is that the identity of the vertices is important, while the theory of circuit partitions of 4-regular multigraphs is independent of the identity of the vertices.

\newcommand{\rotnod}[2]{
\node[scale=1.4, rotate={180+#1}] at ({180+90+#1}:\radius-\labelrad) {{#2}};
}

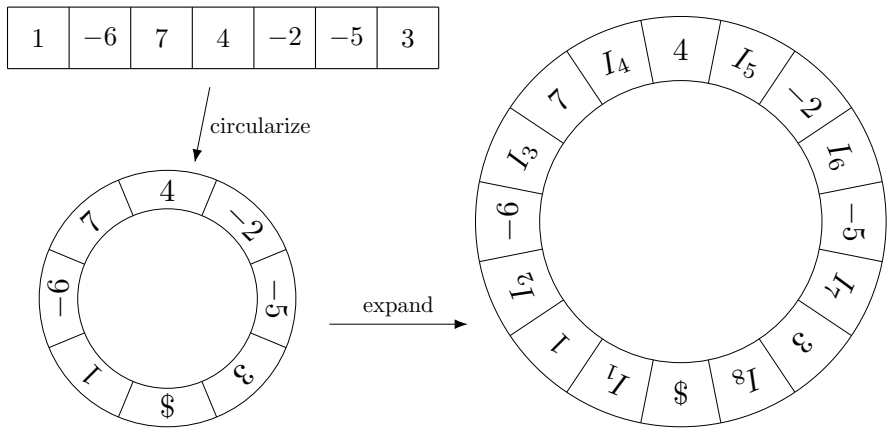
\begin{figure}
\begin{center}
\resizebox{\textwidth}{!}{
\begin{tikzpicture}[scale=4]
\node[scale=1.2] (n1) at (0,0.7) {\signp{1}{-6}{7}{4}{-2}{-5}{3}{0}{0}};
\node (n2) at (-0.1,0.2) {};
\draw[-{Latex[length=2mm]}] (n1) to node [right] {circularize} (n2);
\node (n3) at (0.4,-0.4) {};
\node (n4) at (1.0,-0.4) {};
\draw[-{Latex[length=2mm]}] (n3) to node [above] {expand} (n4);
\begin{scope}[name=scope1,shift={(-0.2,-0.3)}]
\def\radius{0.5cm}
\def\innerradius{0.35cm}
\def\labelrad{0.08cm}
  \draw (0,0) circle (\radius);
  \draw (0,0) circle (\innerradius);
  \rotnod{0}{$\$$}
  \rotnod{-45}{$1$}
  \rotnod{-90}{$-6$}
  \rotnod{-135}{$7$}
  \rotnod{180}{$4$}
  \rotnod{135}{$-2$}
  \rotnod{90}{$-5$}
  \rotnod{45}{$3$}
  \foreach \x in {22.5,67.5,...,360}  \draw (\x:\innerradius) -- (\x:\radius);
\end{scope}
\begin{scope}[name=scope2,shift={(1.8,0)}]
\def\radius{0.8cm}
\def\innerradius{0.55cm}
\def\labelrad{0.13cm}
  \draw (0,0) circle (\radius);
  \draw (0,0) circle (\innerradius);
  \foreach \x in {1,2,...,8}
  {
    \rotnod{22.5-\x*45}{$I_\x$}
  };
  \rotnod{0}{$\$$}
  \rotnod{-45}{$1$}
  \rotnod{-90}{$-6$}
  \rotnod{-135}{$7$}
  \rotnod{180}{$4$}
  \rotnod{135}{$-2$}
  \rotnod{90}{$-5$}
  \rotnod{45}{$3$}
  \foreach \x in {11.25,33.75,...,360}  \draw (\x:\innerradius) -- (\x:\radius);
\end{scope}
\end{tikzpicture}
}
\end{center}
\caption{Circularizing and expanding the signed permutation of Figure~\ref{fig:sign_perm}.}
\label{fig:exp_signed_perm}
\end{figure}


We now describe the construction of the 4-regular multigraph (along with two circuit partitions). As usual, the boundaries between adjacent segments of (the chromosomal depiction of) a signed permutation $\pi$ are called \emph{breakpoints}. Because reversals can also be applied on endpoints of a chromosome, we treat the endpoints of a signed permutation as breakpoints as well. We do this by circularizing the signed permutation, see Figure~\ref{fig:exp_signed_perm}. Note, however, that the location of the endpoints is important. Indeed, e.g., the signed permutation $(1, 2, \cdots, n)$, with $n \geq 2$, has a different reversal distance than any of its proper conjugations (i.e., the signed permutations $(i, i+1, \cdots, n, 1, 2, \cdots, i-1)$ for $i \in \{2,\ldots,n\}$). Therefore, we have anchored the two endpoints to a new segment, which is denoted by $\$$. This corresponds to the usual procedure of \emph{framing} the signed permutation in the theory of sorting by reversals, see, e.g., \cite{dam/Bergeron05}. The next step, which we call here the \emph{expand} step, is to insert an intermediate segment $I_i$ between each two adjacent segments, see again Figure~\ref{fig:exp_signed_perm}.

\newcommand{\rnod}[3]{\node [place] (#1) at (-#2*360-90:\radius) {#3};}
\newcommand{\cordgr}{
\rnod{1}{1/32}{$v_0$}
\rnod{2}{3/32}{$v_0$}
\rnod{3}{5/32}{$v_1$}
\rnod{4}{7/32}{$v_6$}
\rnod{5}{9/32}{$v_5$}
\rnod{6}{11/32}{$v_6$}
\rnod{7}{13/32}{$v_7$}
\rnod{8}{15/32}{$v_3$}
\rnod{9}{17/32}{$v_4$}
\rnod{10}{19/32}{$v_2$}
\rnod{11}{21/32}{$v_1$}
\rnod{12}{23/32}{$v_5$}
\rnod{13}{25/32}{$v_4$}
\rnod{14}{27/32}{$v_2$}
\rnod{15}{29/32}{$v_3$}
\rnod{16}{31/32}{$v_7$}
\draw [arr] (1) to node {$I_1$} (2);
\draw [arr] (3) to node {$I_2$} (4);
\draw [arr] (5) to node {$I_3$} (6);
\draw [arr] (7) to node {$I_4$} (8);
\draw [arr] (9) to node {$I_5$} (10);
\draw [arr] (11) to node {$I_6$} (12);
\draw [arr] (13) to node {$I_7$} (14);
\draw [arr] (15) to node {$I_8$} (16);
\draw [arr] (2) to node {$1$} (3);
\draw [arr] (4) to node {$-6$} (5);
\draw [arr] (6) to node {$7$} (7);
\draw [arr] (8) to node {$4$} (9);
\draw [arr] (10) to node {$-2$} (11);
\draw [arr] (12) to node {$-5$} (13);
\draw [arr] (14) to node {$3$} (15);
\draw [arr] (16) to node {$\$$} (1);
}

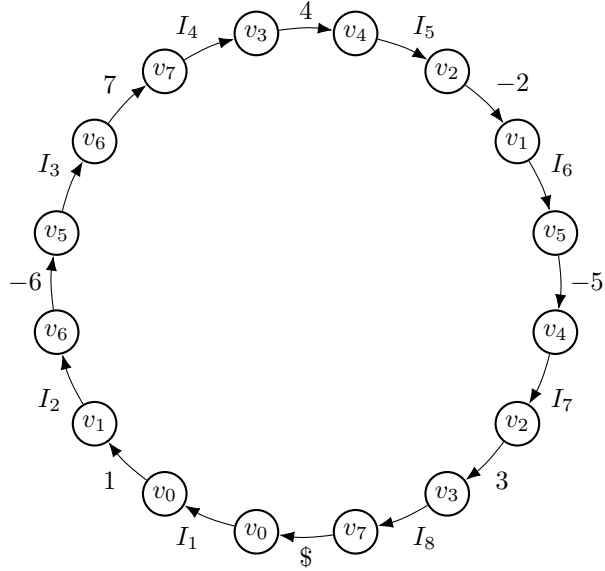
\begin{figure}
\begin{center}
\begin{tikzpicture}[x=1.3cm,y=1.3cm,,place/.style={circle,draw,thick,inner sep=2pt,minimum size=3mm},arr/.style={-{Latex[length=2mm]},bend left=8,auto}]
\def\radius{2.6}
\cordgr
\end{tikzpicture}
\end{center}
\caption{Digraph $D_\pi$ representing the circularized and expanded signed permutation of Figure~\ref{fig:exp_signed_perm}.}
\label{fig:digraph_circ}
\end{figure}

Next, we represent the circularized and expanded signed permutation by a digraph $D_\pi$. In $D_\pi$, each breakpoint is represented by a vertex and each segment is represented by an arrow. The arrow is labeled by the segment $x$ it represents and goes from the left-hand breakpoint of $x$ to the right-hand breakpoint of $x$, see Figure~\ref{fig:digraph_circ}. Moreover, the boundaries/breakpoints of the original segments $i$ and $i+1$ that coincide after the sorting procedure are given a common vertex label $v_i$. For example, the right-hand side breakpoint of segment $3$ and the left-hand side breakpoint of segment $4$ are given a common vertex label $v_3$, see again Figure~\ref{fig:digraph_circ}. Notice that the orientation is important here: the right-hand side breakpoints of segment $1$ and the \emph{right}-hand side breakpoint of segment $-2$ are given a common vertex label because segment $-2$ is segment $2$ in inverted orientation (i.e., rotated by 180 degrees). Since segment $\$$ represents the endpoints, for the arrow corresponding to $\$$, the head vertex is labeled by $v_0$ and the tail vertex is labeled by $v_n$.

\begin{figure}
\begin{center}
\begin{tikzpicture}[auto,x=3cm,y=3cm,every loop/.style={}]
[place/.style={circle,draw,thick,inner sep=0pt,minimum size=3mm}]
\node[place] (n0) at (1.3,1.5) {$v_0$};
\node[place] (n1) at (1.3,2) {$v_1$};
\node[place] (n2) at (3,1) {$v_2$};
\node[place] (n3) at (2,1) {$v_3$};
\node[place] (n4) at (2,0) {$v_4$};
\node[place] (n5) at (-0.5,1) {$v_5$};
\node[place] (n6) at (0.6,1) {$v_6$};
\node[place] (n7) at (1.3,1) {$v_7$};
\draw [bend left=15] (n5) to node {$I_3$} (n6);
\draw [bend right=15] (n5) to node [swap] {$6$} (n6);
\draw [bend left=15] (n2) to node {$I_5$} (n4);
\draw [bend right=15] (n2) to node [swap] {$I_7$} (n4);
\draw [bend left=15] (n1) to node {$2$} (n2);
\draw (n0) to node [swap] {$1$} (n1);
\draw (n0) to node {$\$$} (n7);
\path (n0) edge [loop right] node {$I_1$}();
\draw [bend right=15] (n1) to node [swap] {$I_6$} (n5);
\draw (n1) to node [swap] {$I_2$} (n6);
\draw (n2) to node [swap] {$3$} (n3);
\draw (n3) to node [swap] {$4$} (n4);
\draw [bend right=15] (n3) to node [swap] {$I_8$} (n7);
\draw [bend right=15] (n7) to node [swap] {$I_4$} (n3);
\draw (n7) to node [swap] {$7$} (n6);
\draw [bend left=15] (n4) to node {$5$} (n5);
\end{tikzpicture}
\end{center}
\caption{The $4$-regular multigraph $G_\pi$ of the running example.}
\label{fig:4reg}
\end{figure}
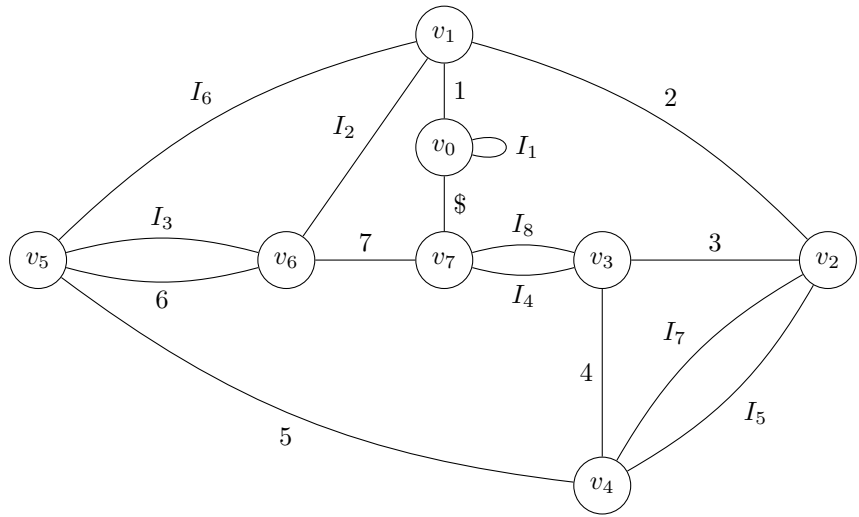

From the digraph $D_\pi$ of Figure~\ref{fig:digraph_circ}, we construct a 4-regular multigraph, denoted by $G_\pi$, by turning each arrow into an undirected edge, removing the signs from the edge labels, and finally merging each two vertices with the same label, see Figure~\ref{fig:4reg}.

Let $G$ be a multigraph and let $l$ be the number of connected components of $G$. A (unoriented) \emph{circuit} of $G$ is a closed walk, without distinguished orientation or starting vertex, allowing repetitions of vertices but not of edges. A \emph{circuit partition} of $G$ is a set $P$ of circuits of $G$ such that each edge of $G$ is in exactly one circuit of $P$. Note that $|P| \geq l$. If $|P| = l$, then we say that $P$ is an \emph{Euler system} of $G$. Note that an Euler system contains an Eulerian circuit (i.e., a circuit visiting each edge exactly once) for each connected component of $G$. In particular, if $G$ is connected, then an Euler system is a singleton containing an Eulerian circuit of $G$.

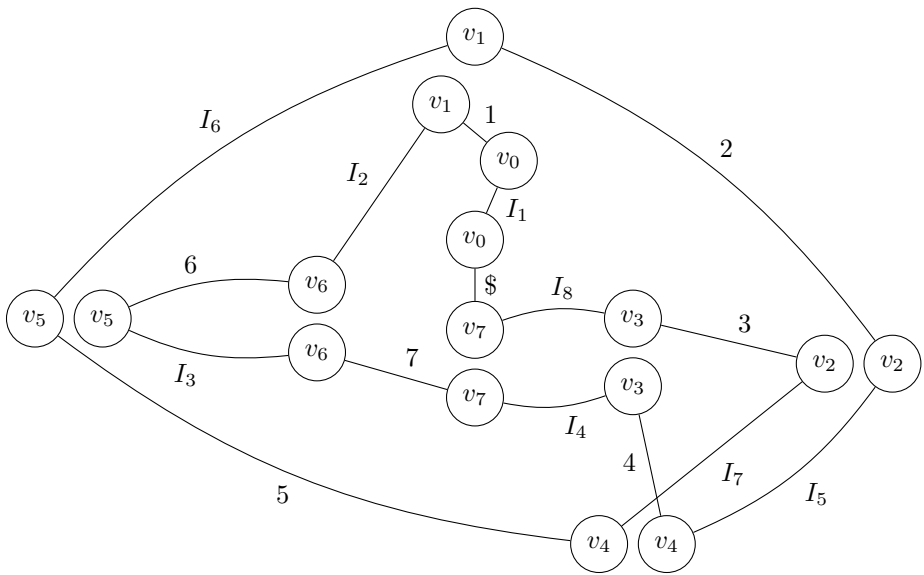
\begin{figure}
\begin{center}
\begin{tikzpicture}[auto,x=3cm,y=3cm,every loop/.style={}]
[place/.style={circle,draw,thick,inner sep=0pt,minimum size=3mm}]
\node[place] (n0a) at (1.3,1.35) {$v_0$};
\node[place] (n0b) at (1.45,1.7) {$v_0$};
\node[place] (n1a) at (1.15,1.95) {$v_1$};
\node[place] (n1b) at (1.3,2.25) {$v_1$};
\node[place] (n2a) at (2.85,0.8) {$v_2$};
\node[place] (n2b) at (3.15,0.8) {$v_2$};
\node[place] (n3a) at (2,1) {$v_3$};
\node[place] (n3b) at (2,0.7) {$v_3$};
\node[place] (n4a) at (1.85,0) {$v_4$};
\node[place] (n4b) at (2.15,0) {$v_4$};
\node[place] (n5a) at (-0.35,1) {$v_5$};
\node[place] (n5b) at (-0.65,1) {$v_5$};
\node[place] (n6a) at (0.6,0.85) {$v_6$};
\node[place] (n6b) at (0.6,1.15) {$v_6$};
\node[place] (n7a) at (1.3,0.95) {$v_7$};
\node[place] (n7b) at (1.3,0.65) {$v_7$};
\draw [bend left=15] (n5a) to node {$6$} (n6b);
\draw [bend right=15] (n5a) to node [swap] {$I_3$} (n6a);
\draw [bend left=15] (n2b) to node {$I_5$} (n4b);
\draw (n2a) to node {$I_7$} (n4a);
\draw [bend left=15] (n1b) to node {$2$} (n2b);
\draw (n0b) to node [swap] {$1$} (n1a);
\draw (n0a) to node {$\$$} (n7a);
\draw (n0b) to node [pos=0.1] {$I_1$} (n0a);
\draw [bend right=15] (n1b) to node [swap] {$I_6$} (n5b);
\draw (n1a) to node [swap] {$I_2$} (n6b);
\draw (n2a) to node [swap] {$3$} (n3a);
\draw (n3b) to node [swap, pos=0.3] {$4$} (n4b);
\draw [bend right=15] (n3a) to node [swap, pos=0.4] {$I_8$} (n7a);
\draw [bend right=15] (n7b) to node [swap] {$I_4$} (n3b);
\draw (n7b) to node [swap] {$7$} (n6a);
\draw [bend left=15] (n4a) to node {$5$} (n5b);
\end{tikzpicture}
\end{center}
\caption{Circuit partition $P_A$ of $G_\pi$ representing the chromosome of species $A$ of Figure~\ref{fig:sign_perm}.}
\label{fig:4reg_splitA}
\end{figure}

\begin{figure}
\begin{center}
\begin{tikzpicture}[auto,x=3cm,y=3cm,every loop/.style={}]
[place/.style={circle,draw,thick,inner sep=0pt,minimum size=3mm}]
\node[place] (n0a) at (1.3,1.5) {$v_0$};
\node[place] (n0b) at (1.6,1.5) {$v_0$};
\node[place] (n1a) at (1.45,2) {$v_1$};
\node[place] (n1b) at (1.15,2) {$v_1$};
\node[place] (n2a) at (3,1) {$v_2$};
\node[place] (n2b) at (3.2,0.8) {$v_2$};
\node[place] (n3a) at (2.2,0.8) {$v_3$};
\node[place] (n3b) at (2,1) {$v_3$};
\node[place] (n4a) at (2,0) {$v_4$};
\node[place] (n4b) at (2.2,-0.2) {$v_4$};
\node[place] (n5a) at (-0.5,0.85) {$v_5$};
\node[place] (n5b) at (-0.5,1.15) {$v_5$};
\node[place] (n6a) at (0.6,0.85) {$v_6$};
\node[place] (n6b) at (0.6,1.15) {$v_6$};
\node[place] (n7a) at (1.15,1) {$v_7$};
\node[place] (n7b) at (1.45,1) {$v_7$};
\draw [bend left=15] (n5b) to node {$I_3$} (n6b);
\draw [bend right=15] (n5a) to node [swap] {$6$} (n6a);
\draw [bend left=15] (n2b) to node {$I_5$} (n4b);
\draw [bend right=15] (n2b) to node [swap] {$I_7$} (n4b);
\draw [bend left=15] (n1a) to node {$2$} (n2a);
\draw (n0a) to node [swap, pos=0.6] {$1$} (n1a);
\draw (n0a) to node [pos=0.2] {$\$$} (n7a);
\path (n0b) edge [loop right] node {$I_1$}();
\draw [bend right=15] (n1b) to node [swap] {$I_6$} (n5b);
\draw (n1b) to node [swap] {$I_2$} (n6b);
\draw (n2a) to node [swap] {$3$} (n3a);
\draw (n3a) to node [swap] {$4$} (n4a);
\draw [bend right=15] (n3b) to node [swap] {$I_8$} (n7b);
\draw [bend right=15] (n7b) to node [swap] {$I_4$} (n3b);
\draw (n7a) to node [swap] {$7$} (n6a);
\draw [bend left=15] (n4a) to node {$5$} (n5a);
\end{tikzpicture}
\end{center}
\caption{Circuit partition $P_B$ of $G_\pi$ containing a circuit $C = (1,2, \cdots,7,\$)$ that represents the chromosome of species $B$ of Figure~\ref{fig:sign_perm}.}
\label{fig:4reg_splitB}
\end{figure}
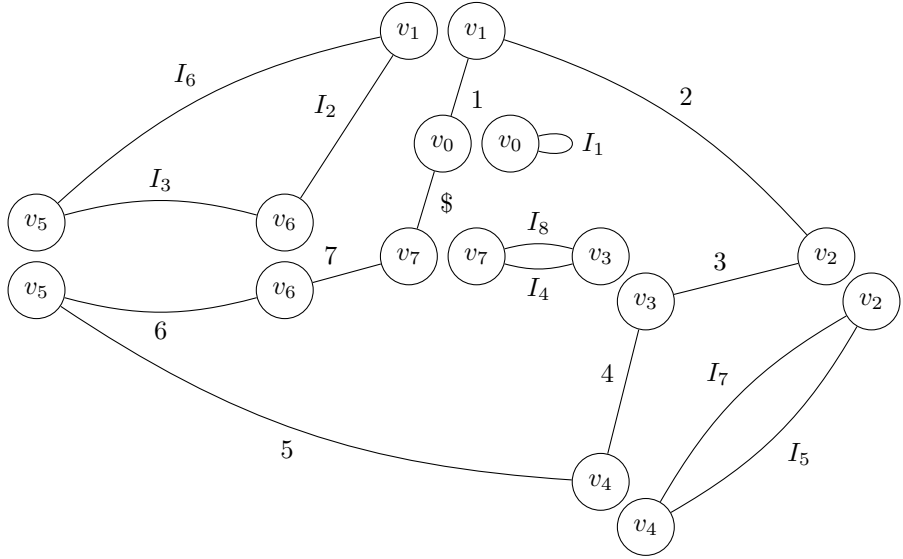

As illustrated in Figure~\ref{fig:4reg_splitA}, the circularized and expanded signed permutation of Figure~\ref{fig:exp_signed_perm} belongs to a particular circuit partition $P_A$ of the 4-regular multigraph $G_\pi$ of Figure~\ref{fig:4reg} (this can also be verified by comparing Figure~\ref{fig:4reg_splitA} with Figure~\ref{fig:digraph_circ}). In this way, $P_A$ is the circuit partition belonging to the chromosome of species $A$. Notice that $P_A$ is an Euler system, in fact, since $G_\pi$ is connected for every signed permutation $\pi$, $P_A$ contains an Eulerian circuit of $G_\pi$. Another circuit partition $P_B$ is illustrated in Figure~\ref{fig:4reg_splitB}. It is the unique circuit partition that includes the circuit $C = (1,2, \cdots,7,\$)$. As such, $P_B$ is the circuit partition belonging to the chromosome of species $B$. Notice that, besides $C$, $P_B$ contains four other circuits in this example. Each of these four circuits consists of intermediate segments (recall that these are the segments of the form $I_i$ for some $i$).

We remark that Figures~\ref{fig:4reg_splitA} and \ref{fig:4reg_splitB} (corresponding to $P_A$ and $P_B$, respectively) can be obtained from $G_\pi$ by ``splitting'' each vertex in an appropriate way. This splitting is considered in \cite{EHG/CyclicGraphDecomp} in the context of gene assembly in ciliates (we recall gene assembly in ciliates in Section~\ref{sec:ga_ciliates}).

While we do not recall the notion of a cycles of a signed permutation (see, e.g., \cite{DBLP:journals/jacm/HannenhalliP99,dam/Bergeron05}), we mention that it is easy to verify that these cycles correspond one-to-one to circuits of intermediate segments of the circuit partition $P_B$.
By using 4-regular multigraphs, we have given these cycles a more ``physical'' interpretation, cf.\ Figure~\ref{fig:4reg_splitB}.

Let $c(\pi)$ be the number of cycles of a signed permutation $\pi$. The following result is well known (in fact, this result has been extended into an equality in \cite{DBLP:journals/jacm/HannenhalliP99}).
\begin{theorem} [\cite{siamcomp/BafnaP96}] \label{thm:rdist_ineq}
Let $\pi$ be a signed permutation with $n$ elements. Then $d_r(\pi) \geq n+1 - c(\pi)$.
\end{theorem}
The proof idea of Theorem~\ref{thm:rdist_ineq} is to show that (1) if $\pi$ is the identity permutation, then $c(\pi)=n+1$ and (2) if $\pi'$ is obtained from $\pi$ by applying a single reversal, then $c(\pi') - c(\pi) \in \{-1,0,1\}$.

We remark that the inequality of Theorem~\ref{thm:rdist_ineq} usually takes the form $d_r(\pi) \geq n' - c(\pi)$, where $n'$ is the number of segments of the framed/anchored signed permutation and is natural when segment $\$$ is instead denoted by $n' = n+1$.

For our running example we see by Figure~\ref{fig:4reg_splitB} that $c(\pi)=4$. Thus $d_r(\pi) \geq 7+1-4=4$. We have seen in Section~\ref{sec:sort_revs}, and in particular Figure~\ref{fig:sign_perm_sort}, that $d_r(\pi) \leq 4$. Consequently, $d_r(\pi)=4$.

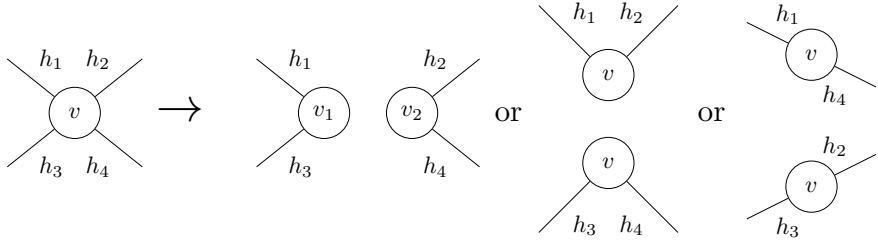
\begin{figure}
\begin{center}
\setlength{\unitlength}{1cm}
\scalebox{0.9}{
\begin{picture}(12.5,3)
\put(-1,0){
\begin{tikzpicture}[auto,x=1.3cm,y=1.3cm]
[place/.style={circle,draw,thick,inner sep=0pt}]
\node[place] (n1) at (1,0.8) {$v$};
\node[place,draw=none] (ntl) at (0,1.6) {};
\node[place,draw=none] (ntr) at (2,1.6) {};
\node[place,draw=none] (nbl) at (0,0) {};
\node[place,draw=none] (nbr) at (2,0) {};
\draw (ntl) to node {$h_1$} (n1);
\draw (ntr) to node [swap] {$h_2$} (n1);
\draw (nbl) to node [swap] {$h_3$} (n1);
\draw (nbr) to node {$h_4$} (n1);
\end{tikzpicture}}
\put(2,1.3){\scalebox{2}{$\rightarrow$}}
\put(2.8,0){\begin{tikzpicture}[auto,x=1.3cm,y=1.3cm]
[place/.style={circle,draw,thick,inner sep=0pt}]
\node[place] (n1) at (1,0.8) {$v_1$};
\node[place] (n2) at (2,0.8) {$v_2$};
\node[place,draw=none] (ntl) at (0,1.6) {};
\node[place,draw=none] (ntr) at (3,1.6) {};
\node[place,draw=none] (nbl) at (0,0) {};
\node[place,draw=none] (nbr) at (3,0) {};
\draw (ntl) to node {$h_1$} (n1);
\draw (ntr) to node [swap] {$h_2$} (n2);
\draw (nbl) to node [swap] {$h_3$} (n1);
\draw (nbr) to node {$h_4$} (n2);
\end{tikzpicture}}
\put(7,1.3){\scalebox{1.3}{or}}
\put(7,-1){\begin{tikzpicture}[auto,x=1.3cm,y=1.3cm]
[place/.style={circle,draw,thick,inner sep=0pt}]
\node[place] (n2) at (1,1) {$v$};
\node[place] (n1) at (1,2) {$v$};
\node[place,draw=none] (ntl) at (0,3) {};
\node[place,draw=none] (ntr) at (2,3) {};
\node[place,draw=none] (nbl) at (0,0) {};
\node[place,draw=none] (nbr) at (2,0) {};
\draw (ntl) to node {$h_1$} (n1);
\draw (ntr) to node [swap] {$h_2$} (n1);
\draw (nbl) to node [swap] {$h_3$} (n2);
\draw (nbr) to node {$h_4$} (n2);
\end{tikzpicture}}
\put(10,1.3){\scalebox{1.3}{or}}
\put(10,-0.7){\begin{tikzpicture}[auto,x=1.3cm,y=1.3cm]
[place/.style={circle,draw,thick,inner sep=0pt}]
\node[place] (n2) at (1,0.5) {$v$};
\node[place] (n1) at (1,2) {$v$};
\node[place,draw=none] (h1) at (0,2.5) {};
\node[place,draw=none] (h4) at (2,1.5) {};
\node[place,draw=none] (h3) at (0,0) {};
\node[place,draw=none] (h2) at (2,1) {};
\draw (h1) to node {$h_1$} (n1);
\draw (h2) to node [swap] {$h_2$} (n2);
\draw (h3) to node [swap] {$h_3$} (n2);
\draw (h4) to node {$h_4$} (n1);
\end{tikzpicture}}
\end{picture}
}
\end{center}
\caption{The three possible routes a circuit partition can take.}
\label{fig:detach}
\end{figure}

Notice that a circuit partition takes, for each vertex of the 4-regular multigraph, one of three possible routes, see Figure~\ref{fig:detach}. Since care must be taken in the case of loops, the four $h_i$'s are not edges but actually ``half-edges'', where two half-edges form an edge. If circuit partitions $P_1$ and $P_2$ take a different route at each vertex of $G$, then we say that $P_1$ and $P_2$ are \emph{supplementary}. Notice that $P_A$ and $P_B$ are supplementary circuit partitions.

For a given circuit partition $P$ and vertex $v$ of $G$, let $P'$ and $P''$ be the circuit partitions obtained from $P$ by changing the route of $P$ at vertex $v$. It is well known that the cardinalities of two of $\{P,P',P''\}$ are equal, to say $k$, and the third is of cardinality $k+1$.

\begin{remark}\label{rem:rev_dist_one_diff_fourreg}
In terms of 4-regular multigraphs, the issue discussed in Remark~\ref{rem:rev_dist_one_diff} translates to the question of whether the anchor should be $\$$ or $-\$$ (in other words, $\$$ in inverted orientation). We assume the former, but the latter anchor is equally valid and may sometimes obtain a reversal distance that is one smaller. We revisit this issue in Section~\ref{sec:DCJ_multiple}.
\end{remark}

\section{Circle graphs} \label{sec:circle_graphs}

\begin{figure}
\begin{center}
\begin{tikzpicture}[x=1.3cm,y=1.3cm,,place/.style={circle,draw,thick,inner sep=2pt,minimum size=3mm},arr/.style={-,bend left=8,auto,text=white,dashed}]
\def\radius{2.6}
\cordgr
\draw [bend right=15] (1) to node {} (2);
\draw (3) to node {} (11);
\draw [bend right=15] (4) to node {} (6);
\draw (5) to node {} (12);
\draw (7) to node {} (16);
\draw (8) to node {} (15);
\draw [bend right=15] (9) to node {} (13);
\draw (10) to node {} (14);
\end{tikzpicture}
\end{center}
\caption{Chord diagram corresponding to the Eulerian circuit depicted in Figure~\ref{fig:4reg_splitA}.}
\label{fig:chord_diagram}
\end{figure}
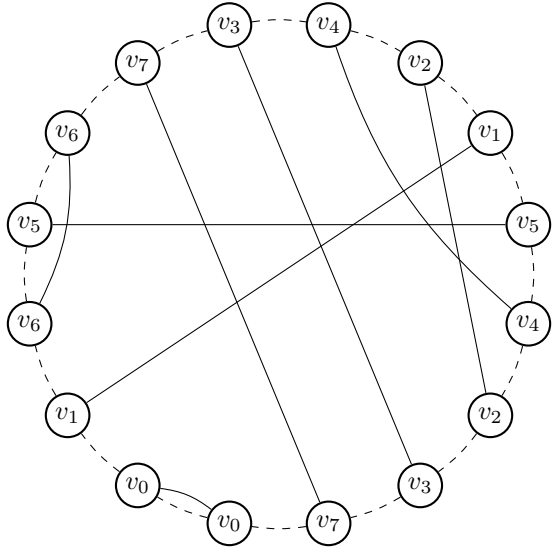

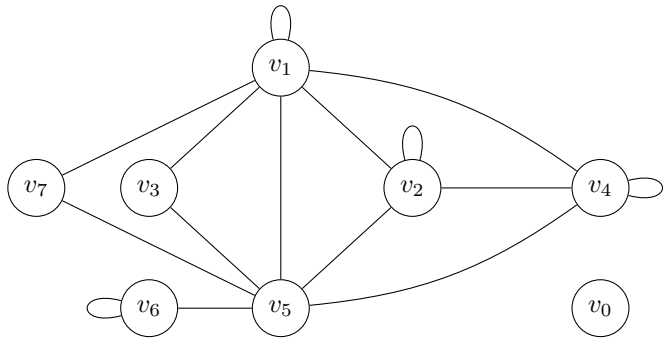
\begin{figure}
\begin{center}
\begin{tikzpicture}[auto,x=2.5cm,y=1.6cm,every loop/.style={}]
[place/.style={circle,draw,thick,inner sep=0pt,minimum size=3mm}]
\node[place] (n0) at (2.7,0) {$v_0$};
\node[place] (n1) at (1,2) {$v_1$};
\node[place] (n2) at (1.7,1) {$v_2$};
\node[place] (n3) at (0.3,1) {$v_3$};
\node[place] (n4) at (2.7,1) {$v_4$};
\node[place] (n5) at (1,0) {$v_5$};
\node[place] (n6) at (0.3,0) {$v_6$};
\node[place] (n7) at (-0.3,1) {$v_7$};
\draw (n1) to node {} (n2);
\draw (n1) to node {} (n3);
\draw [bend left=15] (n1) to node {} (n4);
\draw (n1) to node {} (n5);
\draw (n2) to node {} (n4);
\draw (n2) to node {} (n5);
\draw (n3) to node {} (n5);
\draw [bend left=15] (n4) to node {} (n5);
\draw (n5) to node {} (n6);
\draw (n7) to node {} (n1);
\draw (n7) to node {} (n5);
\path (n1) edge [loop above] node {}();
\path (n2) edge [loop above] node {}();
\path (n4) edge [loop right] node {}();
\path (n6) edge [loop left] node {}();
\end{tikzpicture}
\end{center}
\caption{Circle graph $H_\pi$ for the signed permutation $\pi$ of the running example.}
\label{fig:circle_graph}
\end{figure}

To study the effect of sequences of reversals, we turn to circle graphs. Let us fix two supplementary circuit partitions $P_1$ and $P_2$ of a 4-regular multigraph $G$, where $P_1$ is an Euler system.

A vertex $v$ of $G$ is called \emph{oriented} for $P_1$ with respect to $P_2$ if the circuit partition $P'$ obtained from $P_1$ by changing the route of $P_1$ at $v$ to coincide with the route of $P_2$ at $v$, is an Euler system. We say that vertex $v$ of $G_\pi$ is \emph{oriented} for a signed permutation $\pi$ if $v$ is oriented for $P_A$ with respect to $P_B$. Thus $\{v_1,v_2,v_4,v_6\}$ is the set of oriented vertices of our running example.
%

To construct the circle graph, we assume first, for convenience, that $G$ is connected, i.e., $P_1$ contains a single Eulerian circuit $C$. We draw $C$ as a circle and connect each two vertices with the same label by a chord to obtain a \emph{chord diagram}. See Figure~\ref{fig:chord_diagram} for the chord diagram of the Eulerian circuit of Figure~\ref{fig:4reg_splitA}. We construct a (looped) \emph{circle graph} $H$ for $G$ with respect to $P_1$ and $P_2$ as follows. The set $V(H) = V(G)$ and two distinct vertices of $H$ are adjacent when the corresponding two chords intersect in the chord diagram. Finally, a loop is added for each vertex that is oriented for $P_1$ with respect to $P_2$. In the general case where $G$ is not necessarily connected, the circle graph for $G$ is the union of the circle graphs of each connected component of $G$. The \emph{circle graph} $H_\pi$ for signed permutation $\pi$ is the circle graph of $G_\pi$ with respect to $P_A$ and $P_B$. The circle graph for our running example is depicted in Figure~\ref{fig:circle_graph}.

\begin{remark} \label{rem:overlap_black_white}
We remark that a circle graph is called an ``overlap graph'' in the literature of sorting by reversals. However, we use here the term circle graph because circle graph is the usual name for this notion in mathematics. Also, the vertices of an overlap graph in the literature on sorting by reversals are often decorated by white or black labels instead of loops (black labels correspond to loops). The rest of this section shows why using loops instead of vertex colors is very useful when going to other combinatorial structures and matrices.
\end{remark}

We now recall the well-known notion of an adjacency matrix of a graph. First, the rows and columns of the matrices we consider in this paper are not ordered, but are instead indexed by finite sets $X$ and $Y$, respectively. We call such matrices $X \times Y$-matrices. Note that the usual notions of rank and nullity of such a matrix $A$ are defined --- they are denoted by $r(A)$ and $n(A)$, respectively. The \emph{adjacency matrix} of a graph $G$, denoted by $A(G)$, is the $V(G) \times V(G)$-matrix over the binary field $GF(2)$ where for $v,v' \in V(G)$, the entry indexed by $(v,v')$ is $1$ if and only if $v$ and $v'$ are adjacent (a vertex $v$ is considered adjacent to itself precisely when $v$ has a loop).

The adjacency matrix $A(H_\pi)$ of the circle graph $H_\pi$ of Figure~\ref{fig:circle_graph} is as follows:
$$
A(H_\pi) =
\bordermatrix{
    & v_0 & v_1 & v_2 & v_3 & v_4 & v_5 & v_6 & v_7 \cr
v_0 &   0 &   0 &   0 &   0 &   0 &   0 &   0 &   0 \cr
v_1 &   0 &   1 &   1 &   1 &   1 &   1 &   0 &   1 \cr
v_2 &   0 &   1 &   1 &   0 &   1 &   1 &   0 &   0 \cr
v_3 &   0 &   1 &   0 &   0 &   0 &   1 &   0 &   0 \cr
v_4 &   0 &   1 &   1 &   0 &   1 &   1 &   0 &   0 \cr
v_5 &   0 &   1 &   1 &   1 &   1 &   0 &   1 &   1 \cr
v_6 &   0 &   0 &   0 &   0 &   0 &   1 &   1 &   0 \cr
v_7 &   0 &   1 &   0 &   0 &   0 &   1 &   0 &   0 \cr
}.
$$

We now recall the following result from \cite[Theorem~4]{LT/BinNullity/09}.
\begin{theorem}[\cite{LT/BinNullity/09}]\label{thm:lt_null_circle}
Let $G$ be a 4-regular multigraph with $l$ connected components and let $P_1$ and $P_2$ be supplementary circuit partitions of $G$ with $P_1$ an Euler system. Let $H$ be the circle graph for $G$ with respect to $P_1$ and $P_2$. Then $n(A(H))=|P_2|-c$.
\end{theorem}

We have the following corollary to Theorem~\ref{thm:lt_null_circle}, which considers the case where $H$ is of the form $H_\pi$.
\begin{corollary} \label{cor:nullity_cycles}
Let $\pi$ be a signed permutation. Then $n(A(H_\pi)) = c(\pi)$.
\end{corollary}
\begin{proof}
By Theorem~\ref{thm:lt_null_circle}, $n(A(H_\pi)) = |P_B|-|P_A| = |P_B|-1$ since $P_A$ contains an Eulerian circuit of $G_\pi$. Recall from Section~\ref{sec:4reg_graph} that $c(\pi)$ is the number of circuits of $P_B$ excluding the circuit $(1,\ldots,n,\$)$. Thus $c(\pi) = |P_B|-1$.
\end{proof}
Corollary~\ref{cor:nullity_cycles} illustrates the usefulness of using loops instead of vertex colors for circle graphs (cf.\ Remark~\ref{rem:overlap_black_white}).

For our running example, we have that $c(\pi) = |P_B|-1 = 5-1 = 4$, so $n(A(H_\pi)) = c(\pi) =4$.

Using Corollary~\ref{cor:nullity_cycles}, we can translate the inequality of Theorem~\ref{thm:rdist_ineq} as follows.
\begin{theorem} \label{thm:rdist_ineq_circle}
Let $\pi$ be a signed permutation with $n$ elements. Then $d_r(\pi) \geq r(A(H_\pi))$.
\end{theorem}
\begin{proof}
By Theorem~\ref{thm:rdist_ineq} and Corollary~\ref{cor:nullity_cycles}, $d_r(\pi) \geq n+1 - c(\pi) = n+1 - n(A(H_\pi))$. The result follows by observing that $|V(H_\pi)| = n+1$.
\end{proof}

\begin{lemma} \label{lem:id_perm_zero_rank}
Let $\pi$ be a signed permutation. Then $\pi$ is the identity permutation if and only if $r(A(H_\pi)) = 0$.
\end{lemma}
\begin{proof}
Let $\pi$ have $n$ elements. Note that $\pi$ is the identity permutation if and only if each intermediate segment forms a circuit of length $1$ if and only if $c(\pi) = n+1$. By Corollary~\ref{cor:nullity_cycles}, this is equivalent to $n(A(H_\pi)) = n+1$ and therefore equivalent to $r(A(H_\pi)) = 0$.
\end{proof}

Note that $r(A(H))=0$ simply means that $H$ contains no edges (i.e., consists of only isolated vertices).

\section{Local complementation} \label{sec:lc}

In order to study the effect of reversals on circle graphs, we recall the following graph notions. For a graph $H$ and vertex $v$, the \emph{neighborhood} of $v$ in $H$, denoted by $N_H(v)$, is $\{ v' \in V(G) \mid \{v,v'\} \in E(H), v' \neq v\}$.
\begin{definition}
Let $H$ be a graph and $v$ a looped vertex of $H$. The \emph{local complement} of $H$ at $v$, denoted by $H*v$, is the graph obtained from $H$ by complementing the subgraph induced by $N_H(v)$.\\
In other words, for all $p \subseteq V(H) = V(H*v)$ with $|p| \in \{1,2\}$, we have $p \in E(H*v)$ if and only if either (1) $p \notin E(H)$ and $p \subseteq N_H(v)$ or (2) $p \in E(H)$ and $p \not\subseteq N_H(v)$.
\end{definition}
Moreover, we denote by $H*_c v$ the graph obtained from $H*v$ by removing all edges incident to $v$ (including the loop on $v$). Thus $v$ is an isolated vertex of $H*_c v$. Equivalently, $H *_c v$ complements the subgraph induced by the ``closed neighborhood'' $\{ v' \in V(H) \mid \{v,v'\} \in E(H)\}$. Finally, we denote by $H|v$ the graph obtained from $H*_c v$ by removing the isolated vertex $v$. Note that each of these three operations ($* v$, $*_c v$ and $|v$) are only allowed on looped vertices $v$; we say that such an operation is \emph{applicable} to $H$ if $v$ is a looped vertex of $H$.

The interest of local complement for the topic of sorting by reversals is that it corresponds to applying some particular type of reversal \cite{siamcomp/KaplanST99}.
\begin{theorem}
Let $\pi$ be a signed permutation and let $\pi'$ be the signed permutation obtained from $\pi$ by applying a reversal on the breakpoints corresponding to an oriented vertex $v$. Then $H_{\pi'} = H_\pi *_c v$.
\end{theorem}

It turns out that local complementation has interesting effects on the underlying adjacency matrix (see, e.g., \cite{BH/PivotNullityInvar/09}).
\begin{lemma}\label{lem:null_schur_compl}
Let $H$ be a graph and $v$ a looped vertex of $H$. Then $n(A(H|v)) = n(A(H))$. In other words, $r(A(H *_c v)) = r(A(H))-1$.
\end{lemma}
While Lemma~\ref{lem:null_schur_compl} can be proved directly, another way is to (1) observe that $A(H|v)$ is obtained from $A(H)$ by applying the Schur complementation matrix operation \cite{Schur1917detformula} on the submatrix induced by $\{v\}$ and (2) recall from, e.g., \cite{SchurBook2005} that Schur complementation preserves nullity. We remark that while $A(H|v)$ is obtained from $A(H)$ by applying Schur complementation, $A(H*v)$ is obtained from $A(H)$ by applying the principal pivot transform matrix operation, which is a partial matrix inversion operation, see, e.g., \cite{Tsatsomeros2000151,BH/PivotLoopCompl/09}.

Let $H$ be a graph and let $\sigma = (v_1,\ldots,v_k)$ be a sequence of mutually distinct vertices of $H$.
A sequence $\varphi = *_c v_1 *_c v_2 \cdots *_c v_k$ of $*_c$ operations that is applicable to $H$ (associativity of $*_c$ is from left to right) is called an \emph{lc-sequence} for $H$. We say that an lc-sequence is \emph{full} if $H \varphi$ contains only isolated vertices.

Since (1) by Lemma~\ref{lem:null_schur_compl}, $*_c$ decreases rank by one and (2) a graph $H$ contains only isolated vertices if and only if $r(A(H))=0$, we directly recover the following property observed in \cite[Corollary 4]{LAA/Cooper/2016} (see also \cite[Section~6]{MaxPivotsGraphs/Brijder09}).
\begin{corollary}\label{cor:full_seq_rank}
Let $H$ be a graph and let $\varphi$ be an lc-sequence of length $k$ for $H$. Then $\varphi$ is full if and only if $k = r(A(H))$.
\end{corollary}
Note that if a full lc-sequence $\varphi$ exists for $H_\pi$ with $\pi$ a signed permutation, then by
Theorem~\ref{thm:rdist_ineq_circle} and Lemma~\ref{lem:id_perm_zero_rank} we have that $d_r(\pi) = r(A(H_\pi))$. So, each full lc-sequence corresponds to an optimal sorting of $\pi$.

\section{Hannenhalli-Pevzner theorem} \label{sec:HPthm}

We now recall the so-called Hannenhalli-Pevzner theorem which gives a precise criterion on arbitrary graphs $H$ for the existence of a full lc-sequence. This theorem has been shown in \cite{DBLP:journals/jacm/HannenhalliP99} in terms of signed permutations $\pi$, but has later been extended to arbitrary graphs $H$ (instead of essentially restricting to circle graphs $H_\pi$). Also, this result was shown independently in \cite{SuccessfulnessChar_Original,BinarySymmetric/BrijderH12} in the context of gene assembly in ciliates (we recall this topic in Section~\ref{sec:ga_ciliates}). We give here another proof of this result, closely following the reasoning of \cite{dam/Bergeron05}.

Let $L$ be the set of looped vertices of a graph $H$. For all $v \in V(H)$, we denote $N^{l}_H(v) = N_H(v) \cap L$ and $N^{ul}_H(v) = N_H(v) \setminus L$. Also, $H$ is said to be \emph{loopless} if $L=\emptyset$.
\begin{lemma} \label{lem:subset_contract}
Let $H$ be a connected graph and let $v \in V(H)$ be looped. If $H'$ is a loopless connected component of $H|v$, then both (1) $V(H') \cap N^l_H(v) \neq \emptyset$ and (2) $N^{ul}_H(v) \subseteq N^{ul}_H(w)$ and $N^{l}_H(w) \subseteq N^{l}_H(v)$ for all $w \in V(H') \cap N^l_H(v)$.
\end{lemma}
\begin{proof}
Let $H'$ be a loopless connected component of $H|v$. Since local complementation changes only edges between vertices of $N_H(v)$, $V(H') \cap N_H(v) \neq \emptyset$. Because local complementation complements the loop status of each vertex of $N_H(v)$ and $H'$ is loopless, $V(H') \cap N^l_H(v) = V(H') \cap N_H(v) \neq \emptyset$.

Let $w \in V(H') \cap N^l_H(v)$.

Firstly, let $x \in N^{ul}_H(v)$. Then $x$ is looped in $H|v$. Since $H'$ is loopless, $\{x,w\} \notin E(H|v)$. Thus $\{x,w\} \in E(H)$ (because $x,w \in N_H(v)$). Consequently, $x \in N^{ul}_H(w)$. Thus $N^{ul}_H(v) \subseteq N^{ul}_H(w)$.

Secondly, let $x \in N^{l}_H(w)$. If $x \notin N^{l}_H(v)$, then $x \in N^{l}_{H|v}(w)$ which contradicts the fact that $w$ belongs to a loopless connected component of $H|v$. Thus $N^{l}_H(w) \subseteq N^{l}_H(v)$.
\end{proof}

\newcommand{\maxs}{\mathrm{MS}}

For a vertex $v$ of $H$, define $s(v) = |N^{ul}_H(v)| - |N^{l}_H(v)|$.
Let $\maxs(H)$ be the set of looped vertices $v$ of $H$ such that $s(w) \leq s(v)$ for all $w \in N^l_H(v)$. Note that for any graph $H$ with looped vertices, $\maxs(H)$ is nonempty since it contains all looped vertices that are (globally) maximal with respect to function $s$.


\begin{lemma}\label{lem:contract_not_loopless}
Let $H$ be a connected graph and $v \in \maxs(H)$. Then each loopless connected component of $H|v$ consists of only an isolated vertex.
\end{lemma}
\begin{proof}
Let $H'$ be a loopless connected component of $H|v$. Since $v \in \maxs(H)$, we have by Lemma~\ref{lem:subset_contract}, $N^{ul}_H(v) = N^{ul}_H(w)$ and $N^{l}_H(w) = N^{l}_H(v)$ for some $w \in V(H')$. Since $v$ and $w$ are moreover looped and adjacent, we have that $w$, and therefore $H'$, is an isolated vertex of $H|v$.
\end{proof}

By iteration of Lemma~\ref{lem:contract_not_loopless}, we obtain the following.
\begin{theorem}\label{thm:exists_full_lc_seq}
Let $H$ be a graph. Then there is a full lc-sequence for $H$ if and only if each loopless connected component of $H$ consists of only an isolated vertex.
\end{theorem}

\begin{corollary}[\cite{DBLP:journals/jacm/HannenhalliP99}]
Let $\pi$ be a signed permutation. If each connected component of $H_\pi$ has at least one looped vertex or consists of only an isolated vertex, then $d_r(\pi) = n+1 - c(\pi)$.
\end{corollary}

\section{DCJ operations and multiple chromosomes} \label{sec:DCJ_multiple}

\newcommand{\DNAfragL}[2]{%
\draw #1 -- ++(2,0) -- ++(0,1) -- ++(-2,0);
\draw #1 +(1,0.5) node{#2};
}

\newcommand{\DNAfragR}[2]{%
\draw #1 -- ++(2,0) -- ++(-2,0) -- ++(0,1) -- ++(2,0);
\draw #1 +(1,0.5) node{#2};
}

\newcommand{\DNAintact}[5]{%
\DNAfragL{($#1+(0,2)$)}{#2}
\DNAfragR{($#1+(2,2)$)}{#3}
\DNAfragL{#1}{#4}
\DNAfragR{($#1+(2,0)$)}{#5}
}

\newcommand{\DNAbreak}[5]{%
\DNAfragL{($#1+(0,1.5)$)}{#2}
\DNAfragR{($#1+(2.5,1.5)$)}{#3}
\DNAfragL{#1}{#4}
\DNAfragR{($#1+(2.5,0)$)}{#5}
}

\begin{figure}
\begin{center}
\resizebox{\textwidth}{!}{
\begin{tikzpicture}[auto,x=0.8cm,y=0.8cm]
\coordinate (n1) at (0,0);
\DNAintact{(n1)}{$w$}{$x$}{$y$}{$z$}
\draw[thick,->] (4.4,1.5) -- (5.6,1.5);
\coordinate (n1) at (6,0.25);
\DNAbreak{(n1)}{$w$}{$x$}{$y$}{$z$}
\draw[thick,->] (10.9,1.5) -- (12.1,1.5);
\coordinate (n1) at (12.5,0);
\DNAintact{(n1)}{$w$}{$z$}{$y$}{$x$}
\end{tikzpicture}
}
\end{center}
\caption{The DCJ operation.}
\label{fig:dcj}
\end{figure}
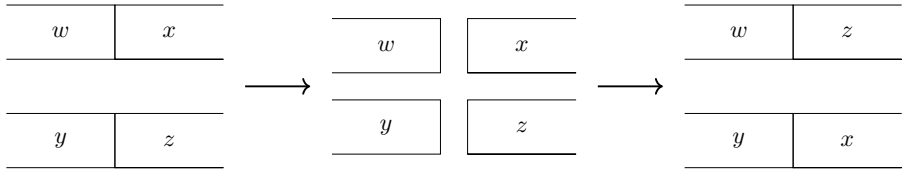

\begin{figure}
\begin{center}
\begin{tikzpicture}[auto,x=0.8cm,y=0.8cm]
\coordinate (n1) at (0,0);
\DNAintact{(n1)}{$w$}{$x$}{$y$}{$z$}
\draw (0,0) -- ++(0,1);
\draw (0,2) -- ++(0,1);
\draw plot [smooth, tension=1] coordinates { (4,2) (5,2) (5,1) (4,1) };
\draw plot [smooth, tension=1] coordinates { (4,3) (6,2.7) (6,0.3) (4,0) };
\end{tikzpicture}
\end{center}
\caption{Alignment corresponding to a reversal.}
\label{fig:reversal}
\end{figure}
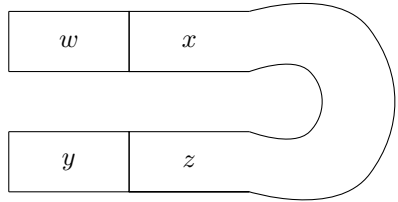

A reversal is a special case of a \emph{double-cut-and-join} (DCJ for short) operation, also called recombination in other contexts. A DCJ operation is depicted in Figure~\ref{fig:dcj}. A DCJ operation consists of three stages: first two distinct breakpoints align, then both breakpoints are cut, and finally the ends are glued back together as depicted in Figure~\ref{fig:dcj}. Since we consider endpoints to be breakpoints as well, any of $w$, $x$, $y$ and $z$ may be nonexistent. From the alignment of Figure~\ref{fig:reversal} one observes that a reversal is a special case of a DCJ operation. DCJ operations are allowed to be intermolecular as well, and so sorting by DCJ operations may involve multiple chromosomes (for example a whole genome) and each chromosome may be linear or circular. The \emph{DCJ distance} of two genomes $g_A$ and $g_B$, denoted by $d_{\mathrm{DCJ}}(g_A,g_B)$, is the minimal number of DCJ operations needed to transform one genome into the other. A toy example of two genomes $g_A$ and $g_B$ of species $A$ and $B$, respectively, consisting of both linear and circular chromosomes is given in Figure~\ref{fig:genomes}.

\newcommand{\signpT}[3]{
\begin{tikzpicture}[auto,x=0.8cm,y=0.8cm]
\draw (0,1) -- (3,1) ;
\draw (0,0) -- (3,0);
\foreach \x in {0,...,3}
{
  \draw (\x,0) -- (\x,1);
}
\draw (0,0) +(0.5,0.5) node{${#1}$};
\draw (1,0) +(0.5,0.5) node{${#2}$};
\draw (2,0) +(0.5,0.5) node{${#3}$};
\end{tikzpicture}
}

\def\radius{0.2cm}
\def\innerradius{0.1cm}
\def\labelrad{0.05cm}
\def\spokes{3}
\newcommand{\circT}{
\scalebox{0.7}{
\begin{tikzpicture}[scale=6]
  \draw (0,0) circle (\radius);
  \draw (0,0) circle (\innerradius);
  \rotnod{0.5*-360/\spokes}{$a$}
  \rotnod{1.5*-360/\spokes}{$-e$}
  \rotnod{2.5*-360/\spokes}{$f$}
  \foreach \x in {30,150,...,360}  \draw (\x:\innerradius) -- (\x:\radius);
\end{tikzpicture}
}}

\def\spokesB{2}
\newcommand{\circD}{
\scalebox{0.7}{
\begin{tikzpicture}[scale=6]
  \draw (0,0) circle (\radius);
  \draw (0,0) circle (\innerradius);
  \rotnod{0.5*-360/\spokesB}{$d$}
  \rotnod{1.5*-360/\spokesB}{$e$}
  \foreach \x in {90,270}  \draw (\x:\innerradius) -- (\x:\radius);
\end{tikzpicture}
}}

\newcommand{\signpO}{
\begin{tikzpicture}[auto,x=0.8cm,y=0.8cm]
\draw (0,1) -- (1,1) ;
\draw (0,0) -- (1,0);
\foreach \x in {0,...,1}
{
  \draw (\x,0) -- (\x,1);
}
\draw (0,0) +(0.5,0.5) node{$f$};
\end{tikzpicture}
}

\begin{figure}
\begin{center}
\resizebox{\textwidth}{!}{
\begin{tikzpicture}
\node (nA1) at (0,0) {\signpT{b}{-d}{c}};
\node (nA2) at (3,0) {\circT};
\draw[dotted] (5,-1.7) -- (5,2.3) ;
\node (nB1) at (8,1.8) {\signpT{a}{b}{c}};
\node (nB2) at (7,0) {\circD};
\node (nB3) at (9,0) {\signpO};
\node (A) at (1.3,-1.5) {genome of species $A$};
\node (B) at (7.8,-1.5) {genome of species $B$};
\end{tikzpicture}
}
\end{center}
\caption{The genomes of species $A$ and $B$.}
\label{fig:genomes}
\end{figure}
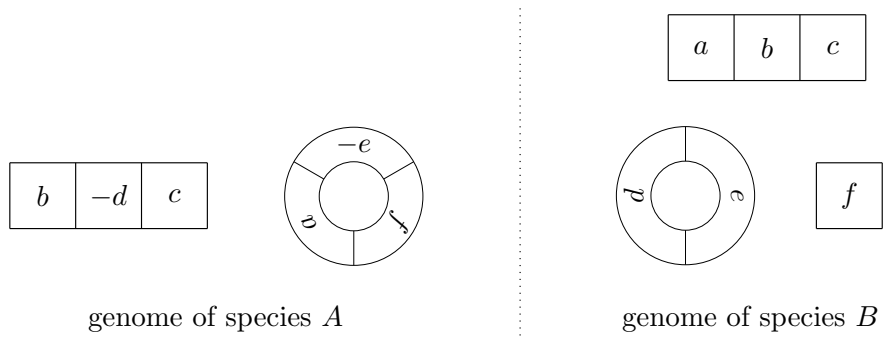

If $g_A$ and $g_B$ contain only circular chromosomes, then the method of Section~\ref{sec:4reg_graph} to construct a 4-regular multigraph applies essentially unchanged --- the only difference is that the circularization step is not done (and so no anchor $\$$ is introduced) because the chromosomes are circular already. Thus, we directly apply the expand step to all chromosomes of $g_A$ and then construct the 4-regular multigraph $G$ as before.
It is now a special case of \cite[Theorem~1]{wabi/BergeronMS06/DCJformula} that $d_{\mathrm{DCJ}}(g_A,g_B) = n-c$, where $n$ is the number of vertices of $G$ (i.e., $n$ is the number of segments in $g_A$ and $g_B$) and $c$ is the number of circuits containing intermediate segments of the circuit partition $P_B$ belonging to $g_B$. The result in \cite{wabi/BergeronMS06/DCJformula} is not stated in terminology of circuit partitions of a 4-regular multigraph, but instead in terms of a graph called an ``adjacency graph''\footnote{The notion of adjacency graph is not to be confused with the different notion of adjacency matrix as recalled earlier in this paper.}. More precisely, in \cite{wabi/BergeronMS06/DCJformula} $c$ is defined as the number of cycles in the adjacency graph, and it can be readily verified that cycles in the adjacency graph correspond one-to-one to circuits of $P_B$ containing intermediate segments.

While the construction of a 4-regular multigraph as outlined in Section~\ref{sec:4reg_graph} works when both genomes have only circular chromosomes, there are issues for genomes containing linear chromosomes such as in Figure~\ref{fig:genomes}. Indeed, the circularization of the linear chromosomes of $g_B$ in Figure~\ref{fig:genomes} leads to two anchors $\$_1$ and $\$_2$ which somehow need to be reconciled with the single anchor $\$$ in the linear chromosome of $g_A$. Also, the issue discussed in Remarks~\ref{rem:rev_dist_one_diff} and \ref{rem:rev_dist_one_diff_fourreg} (which concerns the issue of whether to use $\$$ or $-\$$ as the anchor) is exasperated when there are several linear chromosomes. We leave it as an open problem to resolve this issue of the absence of a canonical 4-regular multigraph for two genomes. We mention that \cite[Theorem~1]{wabi/BergeronMS06/DCJformula} \emph{is} able to calculate the DCJ distance in this general setting (i.e., with linear chromosomes). The formula takes the form $d_{\mathrm{DCJ}}(g_A,g_B) = n-(c+i/2)$, where $i$ is the number of connected components of the adjacency graph that are odd-length paths. Very roughly, one way to explain this formula in terms of 4-regular multigraphs is that each odd-length path corresponds to a side of a linear chromosome and an anchor (which increases $n$ by one) can be introduced in such a way that both sides of a linear chromosome end up in different circuits (which increases $c$ by two). So, the net effect of two sides of a linear chromosome is one, hence the contribution of $i/2$.

Recall that the construction of a circle graph from Section~\ref{sec:circle_graphs} requires two supplementary circuit partitions $P_A$ and $P_B$ of a 4-regular multigraph $G$ where $P_A$ is an Euler system. While in the theory of sorting by reversals $P_A$ is always an Euler system (in fact, $P_A$ contains an Eulerian circuit of $G_\pi$), the circuit partition $P_A$ belonging to a genome $g_A$ is \emph{not} necessarily a Euler system. In Section~\ref{sec:delta_m} we recall the notion of a delta-matroid which can function as a substitute for circle graphs in scenarios where circle graphs do not exist. In fact, as we will see, delta-matroids are useful even in scenarios where circle graphs do exist, such as in the theory of sorting by reversals.

\section{Delta-matroids} \label{sec:delta_m}
It was casually remarked in the discussion of \cite{spire/HartmanV06} that the Hannenhalli-Pevzner theorem might be generalizable to the more general setting of delta-matroids, which are combinatorial structures defined by Bouchet \cite{mp/Bouchet87}. Indeed, we recall now how delta-matroids can be constructed from circuit partitions in 4-regular multigraphs and from circle graphs. In this way, various results concerning the theory of sorting by reversals can be viewed in this more general setting.

A \emph{set system} $D = (V,S)$ is an ordered pair, where $V$ is a finite set, called the \emph{ground set}, and $S$ a set of subsets of $V$. For set systems $D_1 = (V_1,S_1)$ and $D_2 = (V_2,S_2)$ with disjoint ground sets, we define the \emph{direct sum} of $D_1$ and $D_2$, denoted by $D_1 \oplus D_2$, as the set system $(V_1 \cup V_2, \{ X_1 \cup X_2 \mid X_1 \in S_1, X_2 \in S_2 \})$. In this case we say that $D_1$ and $D_2$ are \emph{summands} of $D$. A set system $D$ is called \emph{connected} if it is not the direct sum of two set systems with nonempty ground sets. Also, $D = (V,S)$ is called \emph{even} if the cardinalities of all sets in $S$ are of equal parity. Let us denote symmetric difference by $\sdif$. A \emph{delta-matroid} $D = (V,S)$ is a set system where $S$ is nonempty and, moreover, for all $X,Y \in S$ and $x \in X \sdif Y$, there is an $y \in X \sdif Y$ (possibly equal to $x$) with $X \sdif \{x,y\} \in S$ \cite{mp/Bouchet87}.

Let $G$ be a $4$-regular multigraph and let $P_1$ and $P_2$ be two supplementary circuit partitions of $G$. Denote by $\mathcal{D}_G(P_1,P_2)$ the set system $(V(G),S)$, where for $X \subseteq V(G)$ we have $X \in S$ if and only if the circuit partition $P$ obtained from $P_1$ by changing, for each $v \in X$, the route of $P_1$ at $v$ to coincide with the route of $P_2$ at $v$, is an Euler system.

The following result is stated in \cite[Theorem~5.2]{mp/Bouchet87} (see also \cite[Theorem~5.2]{bouchet1987}).
\begin{theorem} [\cite{mp/Bouchet87}]
Let $G$ be a $4$-regular multigraph and let $P_1$ and $P_2$ be two supplementary circuit partitions of $G$. Then $\mathcal{D}_G(P_1,P_2)$ is a delta-matroid.
\end{theorem}

\begin{example} \label{ex:dm}
Let $G_\pi$ be from Figure~\ref{fig:4reg}, $P_A = \{C\}$ be from Figure~\ref{fig:4reg_splitA} and $P_B$ be from Figure~\ref{fig:4reg_splitB}. Then $\mathcal{D}_{G_\pi}(P_A,P_B) = (V(G),S)$, where
\begin{eqnarray*}
S &=& \{ \emptyset, \{v_1\}, \{v_2\}, \{v_4\}, \{v_6\}, \{v_1,v_3\}, \{v_1,v_5\}, \{v_1,v_6\}, \{v_1,v_7\}, \{v_2,v_4\}, \{v_2,v_5\},\\
 && \quad \{v_2,v_6\}, \{v_3,v_5\}, \{v_4,v_5\}, \{v_4,v_6\}, \{v_5,v_6\}, \{v_5,v_7\}, \{v_1,v_2,v_3\}, \ldots \}.
\end{eqnarray*}
\end{example}

For a graph $H$ and $X \subseteq V(G)$, we denote by $H[X]$ the subgraph of $H$ induced by $X$ (i.e., all vertices outside $X$ are removed including their incident edges). Also, denote by $\mathcal{D}_H$ the set system $(V(H),S)$, where for $X \subseteq V(H)$, $X \in S$ if and only if the matrix $A(H[X])$ is invertible. As usual, the empty matrix is invertible by convention.

\begin{theorem} [Theorem~4.1 of \cite{bouchet1987}]
Let $H$ be a graph. Then $\mathcal{D}_H$ is a delta-matroid.
\end{theorem}
The delta-matroids of the form $\mathcal{D}_H$ for some graph $H$ are called \emph{binary normal} delta-matroids \cite{bouchet1987}.

We now recall the close connection between delta-matroids of circle graphs and of circuit partitions of 4-regular multigraphs.
\begin{theorem} [Theorem~5.3 of \cite{bouchet1987}] \label{thm:dm_circle_CP}
Let $H$ be the circle graph of a 4-regular multigraph $G$ with respect to the supplementary circuit partitions $P_1$ and $P_2$ with $P_1$ an Euler system. Then $\mathcal{D}_H = \mathcal{D}_G(P_1,P_2)$.
\end{theorem}

Note that $G_\pi$, $P_A$ and $P_B$ as used in Example~\ref{ex:dm} are also used to construct the circle graph $H_\pi$ of Figure~\ref{fig:circle_graph}. So, $\mathcal{D}_{G_\pi}(P_A,P_B)$ of Example~\ref{ex:dm} is equal to $\mathcal{D}_{H_\pi}$.

Notice, e.g., that the looped vertices of $H$ precisely correspond to the singletons in $\mathcal{D}_G(P_1,P_2)$. Indeed, the adjacency matrix of a subgraph containing a single vertex $v$ is invertible precisely when $v$ is looped.

By Theorem~\ref{thm:dm_circle_CP}, $\mathcal{D}_G(P_1,P_2)$ corresponds to a circle graph if $P_1$ an Euler system. Recall that in the case of the sorting of genomes by DCJ operations (see Section~\ref{sec:DCJ_multiple}), the usual construction of a circle graph does not work since the input does not necessarily correspond to an Euler system. However, we can construct $\mathcal{D}_G(P_1,P_2)$ and so delta-matroids allow one to work with combinatorial structures similar to circle graphs in cases (such as in Section~\ref{sec:DCJ_multiple}) where a corresponding circle graph does not seem to exist. This is one important reason for considering delta-matroids. Another reason is that delta-matroids are often more easy to work with than circle graphs (even in the cases where circle graphs exist) since the operation of local complementation (on looped vertices) is much more simple in terms of delta-matroids. We will recall this now.

For a set system $D = (V,S)$ and $X \subseteq V$, define the \emph{twist} of $D$ on $X$, denote by $D*X$, as the set system $(V,S')$ with $S' = \{ X \sdif Y \mid Y \in S\}$. Notice that $(D*X)*X = D$. Also, it is easy to verify that if $D$ is a delta-matroid, then so is $D*X$.

\begin{theorem}\label{thm:twist_lc}
Let $H$ be a graph and $v \in V(H)$ be looped. Then $\mathcal{D}_{H*v} = \mathcal{D}_{H}*\{v\}$.
\end{theorem}
By Theorem~\ref{thm:twist_lc}, $\mathcal{D}_{H *_c v}$ is obtained from $\mathcal{D}_{H}*\{v\}$ by removing all sets containing $v$.

We can translate the notion of a full lc-sequence for graphs to the realm of delta-matroids. Let $D = (V,S)$ be a set system. A sequence $\sigma = (v_1,v_2,\ldots,v_k)$ of mutually distinct elements of $V$ is a \emph{lc-sequence} for $D$ if for all $i \in \{0,\ldots,k\}$, $\{v_1,\ldots,v_i\} \in S$. Note that, in particular, $\emptyset, \{v_1,\ldots,v_k\} \in S$. An lc-sequence $\sigma = (v_1,v_2,\ldots,v_k)$ for $D$ is said to be \emph{full} if $\{v_1,\ldots,v_k\} \in \max(S)$, where $\max(S)$ is the set of maximal elements of $S$ with respect to inclusion.

\begin{lemma} \label{lem:lcseq_dm}
Let $H$ be a graph. Then $\sigma = (v_1,v_2,\ldots,v_k)$ is an lc-sequence for $\mathcal{D}_{H}$ if and only if $\varphi = *_c v_1 *_c v_2\cdots *_c v_k$ is an lc-sequence for $H$. Moreover, in this case, $\sigma$ is full if and only if $\varphi$ is full.
\end{lemma}
\begin{proof}
We prove the first statement by induction on $k$. If $k=0$, then $\sigma$ and $\varphi$ are empty sequences and so the result holds trivially. Assume that $k > 0$ and that the result holds for $k' = k-1$.

First, let $\sigma$ be an lc-sequence for $\mathcal{D}_{H}$. In particular, $\sigma' = (v_1,v_2,\ldots,v_{k-1})$ is an lc-sequence for $\mathcal{D}_{H}$. By the induction hypothesis, $\varphi' = *_c v_1 *_c v_2 \cdots *_c v_{k-1}$ is an lc-sequence for $H$. Assume to the contrary that $\varphi$ is not an lc-sequence. Then $v_{k}$ is not a looped vertex of $H \varphi'$. Thus $\{v_{k}\}$ is not a set of $\mathcal{D}_{H \varphi'}$. By the sentence below Theorem~\ref{thm:twist_lc}, $\{v_{k}\}$ is not a set of $\mathcal{D}_{H}*\{v_1,v_2,\ldots,v_{k-1}\}$. In other words, $\{v_1,v_2,\ldots,v_{k}\}$ is not a set of $\mathcal{D}_{H}$ --- a contradiction.

Second, let $\varphi$ be an lc-sequence for $H$. In particular, $\varphi' = *_c v_1 *_c v_2\cdots *_c v_{k-1}$ is an lc-sequence for $H$. By the induction hypothesis, $\sigma' = (v_1,v_2,\ldots,v_{k-1})$ is an lc-sequence for $\mathcal{D}_{H}$. It suffices now to show that $\{v_1,v_2,\ldots,v_{k}\}$ is a set of $\mathcal{D}_{H}$. By the sentence below Theorem~\ref{thm:twist_lc}, $\{v_{k}\}$ is a set of $\mathcal{D}_{H}*\{v_1,v_2,\ldots,v_{k-1}\}$. Thus $\{v_1,v_2,\ldots,v_{k}\}$ is a set of $\mathcal{D}_{H}$.

We now prove the second statement. By the strong principal minor theorem, see \cite{Kodiyalam_Lam_Swan_2008} and also \cite[Lemma 12]{MaxPivotsGraphs/Brijder09}, each set in $\max(S)$ is of cardinality $r(A(H))$. Thus lc-sequence $\sigma$ is full if and only if $k = r(A(H))$ if and only if $\varphi$ is full (by Corollary~\ref{cor:full_seq_rank}).
\end{proof}

We are now ready to rephrase Theorem~\ref{thm:exists_full_lc_seq} in delta-matroid terminology as follows.
\begin{theorem}\label{thm:exists_full_lc_seq_dm}
Let $D$ be a binary normal delta-matroid. Then there is a full lc-sequence for $D$ if and only if each even connected summand of $D$ with nonempty ground set is of the form $(\{v\},\emptyset)$ for some element $v$ in the ground set of $D$.
\end{theorem}
\begin{proof}
Since $D$ is a binary normal delta-matroid, we have that $D = \mathcal{D}_H$ for some graph $H$. By Lemma~\ref{lem:lcseq_dm}, there is a full lc-sequence for $D$ if and only if there is a full lc-sequence for $H$. Also observe that, if a graph $H'$ consists of only an isolated vertex $v$, then $\mathcal{D}_{H'} = (\{v\},\emptyset)$. By Theorem~\ref{thm:exists_full_lc_seq}, it suffices to recall that (1) a graph $H'$ is loopless if and only if $\mathcal{D}_{H'}$ is even (the if-direction is immediate and the only-if direction follows from the well-known fact that invertible zero-diagonal skew-symmetric matrices have even dimensions, see, e.g., \cite{Bouchet_1991_67}) and (2) for all $X \subseteq V(G)$, the subgraph of $H$ induced by $X$ is a (possibly empty) union of connected components of $H$ if and only if there is a summand of $\mathcal{D}_H$ with ground set $X$ \cite[Proposition~5]{DBLP:conf/seccomb/Bouchet91}.
%
\end{proof}
Theorem~\ref{thm:exists_full_lc_seq_dm} does not hold for arbitrary delta-matroids. Indeed, the delta-matroid $D = (V,S)$ with $|V| = 3$ and $S = \{ X \subseteq V \mid |X| \neq 1 \}$ is connected but not even and so the right-hand side of the equivalence of Theorem~\ref{thm:exists_full_lc_seq_dm} trivially holds. However, $D$ does not have a full lc-sequence since $D$ does not contain any singletons. It would be interesting to see to which class of delta-matroids Theorem~\ref{thm:exists_full_lc_seq_dm} can be generalized.

We briefly mention that delta-matroids have been generalized to multimatroids in \cite{DBLP:journals/siamdm/Bouchet97}. In this general setting, delta-matroids translate to a class of multimatroids called 2-matroids. Therefore, this section could also have been phrased in the setting of 2-matroids. Another class of multimatroids, called tight 3-matroids can also be associated to 4-regular multigraphs. See \cite{DBLP:journals/tcs/Brijder15} for the case where the 4-regular multigraph is derived from the context of gene assembly in ciliates, which is a theory closely related to that of sorting by reversals, see Section~\ref{sec:ga_ciliates}. Although it is out of the scope of this paper, it would be interesting to study tight 3-matroids associated to 4-regular multigraphs from the context of sorting by reversals.

\section{A closely related theory: gene assembly in ciliates} \label{sec:ga_ciliates}
Gene assembly is a process taking place during sexual reproduction of unicellular organisms called ciliates. During this process, a nucleus, called the \emph{micronucleus} (or MIC for short), is transformed into another, very different, nucleus, called \emph{macronucleus} (or MAC for short). Each gene in the MAC is one block consisting of a number of consecutive\footnote{Actually, the MDSs overlap slightly in the MAC (the overlapping regions are called pointers), but this is not relevant for this paper.} segments called MDSs. These MDSs also appear in the corresponding MIC gene, but they can appear there in (seemingly) arbitrary order and orientation with respect to the MAC gene and they are moreover separated by noncoding segments called IESs. A (toy) example of a gene consisting of seven MDSs $M_1,\ldots,M_7$ is given in Figures~\ref{fig:MIC} and \ref{fig:MAC}, where the IESs are denoted by $I_1,\ldots,I_8$ and if an MDS $M_i$ in the MIC gene is in inverted orientation (i.e., rotated by 180 degrees) with respect to the MAC gene, then this is denoted by $\overline{M_i}$ (this is the standard notation in this theory, and would of course be written with a minus sign, i.e., $-M_i$, in the theory of sorting by reversals).

The postulated way in which a MIC gene is transformed into its MAC gene is through DCJ operations (called DNA recombination in this context), where MDSs that are not adjacent in the MIC gene but are adjacent in the MAC gene are aligned, cut and glued back such that the two MDSs become adjacent like they appear in the MAC. For example, segments $w$ and $z$ of Figure~\ref{fig:dcj} may be MDSs $M_i$ and $M_{i+1}$, respectively, and $x$ and $y$ IESs. In the \emph{intramolecular} model, the application of DCJ operations is restricted in such a way that they cannot result in MDSs appearing in different molecules \cite{GeneAssemblyBook} (however, it is allowed to apply two DCJ operations simultaneously, when each of them separately would result in a split of the molecule).

Starting from the MIC gene, Figure~\ref{fig:MIC}, we can construct a 4-regular multigraph in a similar way as for the theory of sorting by reversals. However, unlike before, the left-hand side of the first MDS $M_1$ and the right-hand side of the last MDS $M_7$ are not considered breakpoints. Hence, in Figure~\ref{fig:MIC}, we merge adjacent segments $I_1$ and $M_1$ ($M_7$ and $I_4$, respectively) in the MIC gene into a single segment called $I_1 M_1$ ($M_7 I_4$, respectively). Moreover, recall that before we introduced an anchor segment $\$$ during the circularization, because the endpoints of a chromosome are breakpoints too. In the case of MIC genes, there are no endpoints that are breakpoints and so we omit the anchor segment $\$$. Instead, during the circularization the rightmost segment $I_8$ is merged with the leftmost segment $I_1 M_1$ in the MIC gene into a single segment called $I_8;I_1M_1$, see Figure~\ref{fig:circ_ga}. Finally, because the ``signed permutation'' of Figure~\ref{fig:MIC} already contains intermediate segments (the IESs), we also do not have an expand step like in Figure~\ref{fig:exp_signed_perm}. In this way, the intermediate segments are physical, while they are ``virtual'' in the theory of sorting by reversals.

The construction of a 4-regular multigraph from Figure~\ref{fig:circ_ga} is now identical as in the theory of sorting by reversals. Due the similarity of the ``signed permutations'' of Figures~\ref{fig:sign_perm} and \ref{fig:MIC}, we see that the circle graph corresponding to Figure~\ref{fig:circ_ga} is obtained from the circle graph of Figure~\ref{fig:circle_graph} by removing vertices $v_0$ and $v_7$. See \cite{BH/algebra-Tampa} for more details on how the 4-regular multigraphs, circle graphs and delta-matroids can be used to study gene assembly in ciliates.

From the discussion above it is not surprising that the theory of sorting by reversals and the theory of gene assembly in ciliates partly overlap. For example, in both theories local complementation plays an important role --- indeed, as we mentioned in Section~\ref{sec:HPthm} the Hannenhalli-Pevzner theorem was discovered independently in both theories. We remark that links between gene assembly in ciliates and sorting by reversals have also been appreciated in \cite{Extended_paper,involve/ciliateSortingReversal,ciliateSortingReversal}.

\begin{figure}
\begin{center}
\resizebox{\textwidth}{!}{
\begin{tikzpicture}[auto,x=0.8cm,y=0.8cm]
\draw (0,1) -- (15,1);
\draw (0,0) -- (15,0);
\foreach \x in {1,...,14}
{
  \draw (\x,0) -- (\x,1);
}
\draw (0,0) +(0.5,0.5) node{$I_1$};
\draw (1,0) +(0.5,0.5) node{$M_{1}$};
\draw (2,0) +(0.5,0.5) node{$I_2$};
\draw (3,0) +(0.5,0.5) node{$\overline{M_{6}}$};
\draw (4,0) +(0.5,0.5) node{$I_3$};
\draw (5,0) +(0.5,0.5) node{$M_7$};
\draw (6,0) +(0.5,0.5) node{$I_4$};
\draw (7,0) +(0.5,0.5) node{$M_4$};
\draw (8,0) +(0.5,0.5) node{$I_5$};
\draw (9,0) +(0.5,0.5) node{$\overline{M_{2}}$};
\draw (10,0) +(0.5,0.5) node{$I_6$};
\draw (11,0) +(0.5,0.5) node{$\overline{M_{5}}$};
\draw (12,0) +(0.5,0.5) node{$I_7$};
\draw (13,0) +(0.5,0.5) node{$M_{3}$};
\draw (14,0) +(0.5,0.5) node{$I_8$};
\end{tikzpicture}
}
\end{center}
\caption{MIC form of a gene}
\label{fig:MIC}
\end{figure}

\begin{figure}
\begin{center}
\begin{tikzpicture}[auto,x=0.8cm,y=0.8cm]
\draw (0,1) -- (9,1);
\draw (0,0) -- (9,0);
\foreach \x in {1,...,8}
{
  \draw (\x,0) -- (\x,1);
}
\foreach \x in {1,...,7}
{
  \draw (\x,0) +(0.5,0.5) node{$M_{\x}$};
}
\end{tikzpicture}
\end{center}
\caption{MAC form of a gene}
\label{fig:MAC}
\end{figure}


\def\radius{0.8cm}
\def\innerradius{0.65cm}
\def\labelrad{0.08cm}
\def\spokes{12}
\begin{figure}
\begin{center}
\scalebox{0.7}{
\begin{tikzpicture}[scale=6]
  \draw (0,0) circle (\radius);
  \draw (0,0) circle (\innerradius);
  \rotnod{0*-360/\spokes}{$I_8 ; I_1 M_1$}
  \rotnod{1*-360/\spokes}{$I_2$}
  \rotnod{2*-360/\spokes}{$\overline{M_6}$}
  \rotnod{3*-360/\spokes}{$I_3$}
  \rotnod{4*-360/\spokes}{$M_7 I_4$}
  \rotnod{5*-360/\spokes}{$M_4$}
  \rotnod{6*-360/\spokes}{$I_5$}
  \rotnod{7*-360/\spokes}{$\overline{M_2}$}
  \rotnod{8*-360/\spokes}{$I_6$}
  \rotnod{9*-360/\spokes}{$\overline{M_5}$}
  \rotnod{10*-360/\spokes}{$I_7$}
  \rotnod{11*-360/\spokes}{$M_3$}
  \foreach \x in {15,45,...,360}  \draw (\x:\innerradius) -- (\x:\radius);
\end{tikzpicture}
}
\end{center}
\caption{Circularization of the MIC gene of Figure~\ref{fig:MIC}.}
\label{fig:circ_ga}
\end{figure}
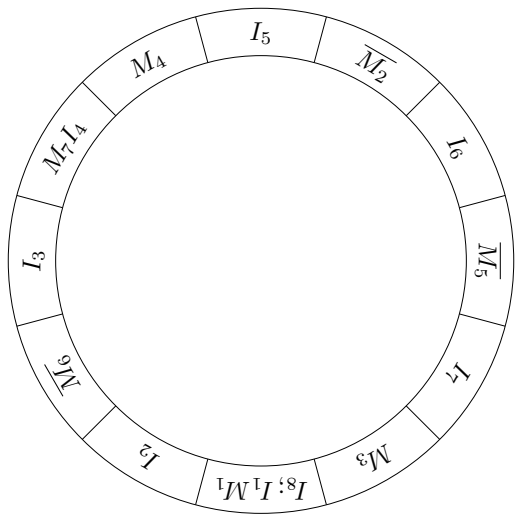

\section{Discussion} \label{sec:disc}
We have shown that the theory of 4-regular multigraphs, including their accompanying combinatorial structures such as circle graphs and delta-matroids, can be applied to the topic of sorting by reversals. The fact that a notion such as local complementation has been rediscovered in the context of sorting by reversals signifies the importance of the theory of 4-regular multigraphs.

This paper may serve as an introduction to the theory of 4-regular multigraphs for the audience familiar with sorting by reversals. This paper has covered only very little of the extensive body of knowledge that the theory of 4-regular multigraphs provides and that can be applied to the topic of sorting by reversals. Indeed, we also mention for example the notion of \emph{touch graph} \cite[Section~6]{Bouchet/87/ejc/isotropicsys} (see also, e.g., \cite{EUJC/Traldi/transmat}) of a circuit partition $P$ that is likely to be useful for sorting by reversals. Indeed, while omitting details, the touch graph of a circuit partition $P_B$ belonging to species $B$ in the context of sorting by reversals (and also in the context of gene assembly in ciliates) is of a very special form, that of a star graph, which likely have interesting consequences. Indeed, while only implicitly stated in \cite{OverlapRedGrFull/Brijder07}, this star graph property is key in the main result of \cite{OverlapRedGrFull/Brijder07} from the context of gene assembly in ciliates.

Since the theory of gene assembly in ciliates can also be fit into the theory of 4-regular multigraphs, we have also extended the links between the theories of gene assembly in ciliates and  sorting by reversals as observed in \cite{Extended_paper,involve/ciliateSortingReversal,ciliateSortingReversal}.

Finally, we have formulated the open problem of using the theory of 4-regular multigraphs in the case of DCJ operations in the presence of linear chromosomes (see Section~\ref{sec:DCJ_multiple}) and also the open problem of generalizing the Hannenhalli-Pevzner theorem from graphs to a suitable subclass of delta-matroids (see Section~\ref{sec:delta_m}).

\subsection*{Acknowledgements}
We thank anonymous referees for their helpful comments on an earlier version of this paper.

\bibliography{../mmatroids}

\begin{thebibliography}{10}

\bibitem{Abrham1980/EulerTours}
J.~Abrham and A.~Kotzig.
\newblock Transformations of {Euler} tours.
\newblock In M.~Deza and I.~G. Rosenberg, editors, {\em Combinatorics 79 Part
  I}, volume~8 of {\em Annals of Discrete Mathematics}, pages 65--69. Elsevier,
  1980.

\bibitem{siamcomp/BafnaP96}
V.~Bafna and P.~A. Pevzner.
\newblock Genome rearrangements and sorting by reversals.
\newblock {\em SIAM Journal on Computing}, 25(2):272--289, 1996.

\bibitem{dam/Bergeron05}
A.~Bergeron.
\newblock A very elementary presentation of the {Hannenhalli-Pevzner} theory.
\newblock {\em Discrete Applied Mathematics}, 146(2):134--145, 2005.

\bibitem{wabi/BergeronMS06/DCJformula}
A.~Bergeron, J.~Mixtacki, and J.~Stoye.
\newblock A unifying view of genome rearrangements.
\newblock In P.~Bucher and B.~M.~E. Moret, editors, {\em Proceedings of the 6th
  International Workshop on Algorithms in Bioinformatics (WABI 2006)}, volume
  4175 of {\em Lecture Notes in Computer Science}, pages 163--173. Springer,
  2006.

\bibitem{mp/Bouchet87}
A.~Bouchet.
\newblock Greedy algorithm and symmetric matroids.
\newblock {\em Mathematical Programming}, 38(2):147--159, 1987.

\bibitem{Bouchet/87/ejc/isotropicsys}
A.~Bouchet.
\newblock Isotropic systems.
\newblock {\em European Journal of Combinatorics}, 8(3):231--244, 1987.

\bibitem{bouchet1987}
A.~Bouchet.
\newblock Representability of {$\Delta$}-matroids.
\newblock In {\em Proceedings of the 6th Hungarian Colloquium of Combinatorics,
  Colloquia Mathematica Societatis J\'{a}nos Bolyai}, volume~52, pages
  167--182. North-Holland, 1987.

\bibitem{Bouchet1989/maps_deltam}
A.~Bouchet.
\newblock Maps and {$\Delta$}-matroids.
\newblock {\em Discrete Mathematics}, 78(1-2):59--71, 1989.

\bibitem{DBLP:conf/seccomb/Bouchet91}
A.~Bouchet.
\newblock Matroid connectivity and fundamental graphs.
\newblock In {\em Proceedings of the 22nd Southeastern Conference on
  Combinatorics, Graph Theory, and Computing}, volume~85 of {\em Congressus
  Numerantium}, pages 81--88, 1991.

\bibitem{DBLP:journals/siamdm/Bouchet97}
A.~Bouchet.
\newblock Multimatroids {I.} {Coverings} by independent sets.
\newblock {\em SIAM Journal on Discrete Mathematics}, 10(4):626--646, 1997.

\bibitem{Bouchet_1991_67}
A.~Bouchet and A.~Duchamp.
\newblock Representability of {$\Delta$}-matroids over {$GF(2)$}.
\newblock {\em Linear Algebra and its Applications}, 146:67--78, 1991.

\bibitem{DBLP:journals/tcs/Brijder15}
R.~Brijder.
\newblock Recombination faults in gene assembly in ciliates modeled using
  multimatroids.
\newblock {\em Theoretical Computer Science}, 608:27--35, 2015.

\bibitem{MaxPivotsGraphs/Brijder09}
R.~Brijder and H.~J. Hoogeboom.
\newblock Maximal pivots on graphs with an application to gene assembly.
\newblock {\em Discrete Applied Mathematics}, 158(18):1977--1985, 2010.

\bibitem{BH/PivotLoopCompl/09}
R.~Brijder and H.~J. Hoogeboom.
\newblock The group structure of pivot and loop complementation on graphs and
  set systems.
\newblock {\em European Journal of Combinatorics}, 32:1353--1367, 2011.

\bibitem{BH/PivotNullityInvar/09}
R.~Brijder and H.~J. Hoogeboom.
\newblock Nullity invariance for pivot and the interlace polynomial.
\newblock {\em Linear Algebra and its Applications}, 435:277--288, 2011.

\bibitem{BinarySymmetric/BrijderH12}
R.~Brijder and H.~J. Hoogeboom.
\newblock Binary symmetric matrix inversion through local complementation.
\newblock {\em Fundamenta Informaticae}, 116(1-4):15--23, 2012.

\bibitem{BH/algebra-Tampa}
R.~Brijder and H.~J. Hoogeboom.
\newblock The algebra of gene assembly in ciliates.
\newblock In N.~Jonoska and M.~Saito, editors, {\em Discrete and Topological
  Models in Molecular Biology}, Natural Computing Series, pages 289--307.
  Springer, 2014.

\bibitem{Extended_paper}
R.~Brijder, H.~J. Hoogeboom, and G.~Rozenberg.
\newblock Reducibility of gene patterns in ciliates using the breakpoint graph.
\newblock {\em Theoretical Computer Science}, 356:26--45, 2006.

\bibitem{OverlapRedGrFull/Brijder07}
R.~Brijder, H.~J. Hoogeboom, and G.~Rozenberg.
\newblock Reduction graphs from overlap graphs for gene assembly in ciliates.
\newblock {\em International Journal of Foundations of Computer Science},
  20:271--291, 2009.

\bibitem{ChunMoffatt/DeltaM/EmbeddedGraphs}
C.~Chun, I.~Moffatt, S.~Noble, and R.~Rueckriemen.
\newblock Matroids, delta-matroids and embedded graphs.
\newblock [arXiv:1403.0920], 2014.

\bibitem{LAA/Cooper/2016}
J.~Cooper and J.~Davis.
\newblock Successful pressing sequences for a bicolored graph and binary
  matrices.
\newblock {\em Linear Algebra and its Applications}, 490:162--173, 2016.

\bibitem{GeneAssemblyBook}
A.~Ehrenfeucht, T.~Harju, I.~Petre, D.~M. Prescott, and G.~Rozenberg.
\newblock {\em Computation in Living Cells -- Gene Assembly in Ciliates}.
\newblock Springer Verlag, 2004.

\bibitem{SuccessfulnessChar_Original}
A.~Ehrenfeucht, T.~Harju, I.~Petre, and G.~Rozenberg.
\newblock Characterizing the micronuclear gene patterns in ciliates.
\newblock {\em Theory of Computing Systems}, 35:501--519, 2002.

\bibitem{EHG/CyclicGraphDecomp}
A.~Ehrenfeucht, T.~Harju, and G.~Rozenberg.
\newblock Gene assembly through cyclic graph decomposition.
\newblock {\em Theoretical Computer Science}, 281:325--349, 2002.

\bibitem{genomeRearr/Fertin/2009}
G.~Fertin, A.~Labarre, I.~Rusu, E.~Tannier, and S.~Vialette.
\newblock {\em Combinatorics of Genome Rearrangements}.
\newblock MIT Press, 2009.

\bibitem{Fleischner/EulerianTrails}
H.~Fleischner, G.~Sabidussi, and E.~Wenger.
\newblock Transforming eulerian trails.
\newblock {\em Discrete Mathematics}, 109(1):103--116, 1992.

\bibitem{DBLP:journals/jacm/HannenhalliP99}
S.~Hannenhalli and P.~A. Pevzner.
\newblock Transforming cabbage into turnip: Polynomial algorithm for sorting
  signed permutations by reversals.
\newblock {\em Journal of the ACM}, 46(1):1--27, 1999.

\bibitem{spire/HartmanV06}
T.~Hartman and E.~Verbin.
\newblock Matrix tightness: {A} linear-algebraic framework for sorting by
  transpositions.
\newblock In F.~Crestani, P.~Ferragina, and M.~Sanderson, editors, {\em
  Proceedings of the 13th International Conference on String Processing and
  Information Retrieval (SPIRE 2006)}, volume 4209 of {\em Lecture Notes in
  Computer Science}, pages 279--290. Springer, 2006.

\bibitem{involve/ciliateSortingReversal}
J.~L. Herlin, A.~Nelson, and M.~Scheepers.
\newblock Using ciliate operations to construct chromosome phylogenies.
\newblock {\em Involve}, 9(1):1--26, 2016.

\bibitem{ciliateSortingReversal}
C.~L. Jansen, M.~Scheepers, S.~L. Simon, and E.~Tatum.
\newblock Context directed reversals and the ciliate decryptome.
\newblock [arXiv:1603.06149v3], 2016.

\bibitem{siamcomp/KaplanST99}
H.~Kaplan, R.~Shamir, and R.~E. Tarjan.
\newblock A faster and simpler algorithm for sorting signed permutations by
  reversals.
\newblock {\em {SIAM} Journal on Computing}, 29(3):880--892, 1999.

\bibitem{Kodiyalam_Lam_Swan_2008}
V.~Kodiyalam, T.~Lam, and R.~Swan.
\newblock Determinantal ideals, {Pfaffian} ideals, and the principal minor
  theorem.
\newblock In {\em Noncommutative Rings, Group Rings, Diagram Algebras and Their
  Applications}, pages 35--60. American Mathematical Society, 2008.

\bibitem{kotzig1968}
A.~Kotzig.
\newblock Eulerian lines in finite 4-valent graphs and their transformations.
\newblock In {\em Theory of graphs, Proceedings of the Colloquium, Tihany,
  Hungary, 1966}, pages 219--230. Academic Press, New York, 1968.

\bibitem{PevznerBook}
P.~A. Pevzner.
\newblock {\em Computational Molecular Biology: An Algorithmic Approach}.
\newblock MIT Press, 2000.

\bibitem{Schur1917detformula}
J.~Schur.
\newblock {\"U}ber {P}otenzreihen, die im {I}nnern des {E}inheitskreises
  beschr{\"a}nkt sind.
\newblock {\em Journal f{\"u}r die reine und angewandte Mathematik},
  147:205--232, 1917.

\bibitem{LT/BinNullity/09}
L.~Traldi.
\newblock Binary nullity, {Euler} circuits and interlace polynomials.
\newblock {\em European Journal of Combinatorics}, 32(6):944--950, 2011.

\bibitem{EUJC/Traldi/transmat}
L.~Traldi.
\newblock The transition matroid of a 4-regular graph: An introduction.
\newblock {\em European Journal of Combinatorics}, 50:180 -- 207, 2015.

\bibitem{Tsatsomeros2000151}
M.~Tsatsomeros.
\newblock Principal pivot transforms: properties and applications.
\newblock {\em Linear Algebra and its Applications}, 307(1-3):151--165, 2000.

\bibitem{SchurBook2005}
F.~Zhang, editor.
\newblock {\em The Schur Complement and Its Applications}, volume~4 of {\em
  Numerical Methods and Algorithms}.
\newblock Springer, 2005.

\end{thebibliography}

\end{document}